\documentclass{cmslatex}
\usepackage[paperwidth=7in, paperheight=10in, margin=.875in]{geometry}
 \usepackage[pdftex,
breaklinks,
bookmarks,
pdfsubject={LaTeX},
colorlinks,
linkcolor=red,
anchorcolor=green,
urlcolor=blue,
citecolor=blue]{hyperref}

\usepackage{amsfonts,amssymb}
\usepackage{amsmath}
\usepackage{graphicx}
\usepackage{enumerate}
\renewcommand{\theequation}{\arabic{section}.\arabic{equation}}

\usepackage{mathtools}
\usepackage{mathrsfs} %math symbol caligraphic
\usepackage{IEEEtrantools} %for equations
\usepackage{fp}   %to allow plain calculation in latex 

\usepackage[subnum]{cases} % to use system equation numbering, fleqn = align left by default
\usepackage{cancel}

\usepackage{xstring} %for \IfEqCase

\def\u2#1{\underline{\underline{#1}}} % underline twice
\newcommand{\bs}[1]{\boldsymbol{#1}} % boldsymbol shortcut
\newcommand{\subscript}[2]{$#1 _ #2$} % for subscript in enumerate

\usepackage[shortlabels]{enumitem} %shortlabel option needed to make enumerate and enumitem compatible
\usepackage{pifont}     % pour les items

\usepackage{parskip} %to jump line btw paragraphe
\usepackage[normalem]{ulem} % to strike out text with \sout{} command
\usepackage{accents}
\usepackage{csquotes} %the quoted text is typeset in accordance with the rules of the language
\def\colorize<#1>{\temporal<#1>% %%% In "itemize" environment, show
{\color{black!30}}%   next lines in grey
{\color{red}}%      current lines in red
{\color{black}}%    previous lines in black
}
\usepackage{varwidth,lipsum} %sim­i­lar to mini­page, but the spec­i­fied width is just a max­i­mum value
\usepackage{tabto} %pour des tabs
\usepackage{tabstackengine} %to tab in text

\graphicspath{{image/}}

\usepackage[
style=numeric-comp,
citestyle=numeric,
backend=biber,
firstinits=true,%   initiales pour les pr\'enoms
useprefix=true,%         pour faire apparaitre ``van Leer'' et pas ``Leer'' dans les citations.
sorting=none, %nty, %name title and year
sortcites=true,
bibencoding=utf-8,
maxnames=15,
minnames=1,
maxbibnames=15,%         nombre maximum d'auteurs cit\'es avant de passer à la forme ``et.al'' DANS LA BIBLIO
minbibnames=14,
maxcitenames=2,%         nombre maximum d'auteurs cit\'es avant de passer à la forme ``et.al'' DANS UNE CITATIO
mincitenames=2,
hyperref=auto,
backref=false,
backrefstyle=all+,
isbn=false,
eprint=false,
url=false,
doi=false%
]{biblatex}

\renewbibmacro*{volume+number+eid}{%
  \printfield{volume}%
  \printfield{number}%
  \printfield{eid}}

\DeclareNameAlias{sortname}{first-last}
\renewbibmacro{in:}{}

\DeclareFieldFormat[article]{volume}{#1}% suppress prefix volume of a journal
\DeclareFieldFormat[article]{number}{#1}% suppress prefix number of a journal
\DeclareFieldFormat[article]{title}{\mkbibemph{#1}} %remove the quote of the default title style ITALIC
\DeclareFieldFormat{journaltitle}{#1} %make journal title not italic
\DeclareFieldFormat[article]{pages}{#1}             % no prefix for pages but here ":"
\DeclareFieldFormat[article]{year}{#1}

\renewbibmacro*{journal+issuetitle}{%
  \setunit{\addcomma\space \addspace}% to have comma and space bfr journal name
  \printfield{\textbf{journaltitle}} %italic 
\newunit}

\DeclareBibliographyDriver{article}{%
  \usebibmacro{bibindex}%
  \usebibmacro{begentry}%
  \usebibmacro{author/translator+others}%
  \setunit{\addcomma\space }\newblock %separator author names and title : here it is a dot and a coma
  \setunit{\addcomma\space} % separator btw author and title
  \printfield{title}
  \setunit{\addcomma\space} % separator btw title and journaltitle
  \printfield{journaltitle}
  \setunit{\addcomma\space} % separator btw title and journaltitle
  \setunit{\addcomma\space}\newblock  %comma after journal
  \iffieldundef{volume}{
  }
  {\iffieldundef{number}{\printfield{volume}
  }
  {\printfield{volume}(\printfield{number})}}
  \setunit{\addcolon}
  \printfield{pages}
  \setunit{\addcomma\space}\newblock  %comma after journal
  \printfield{year}
  \usebibmacro{finentry}}

\DeclareFieldFormat[thesis]{author}{#1}% suppress prefix volume of a journal
\DeclareFieldFormat[thesis]{title}{\mkbibemph{#1}}% suppress prefix volume of a journal
\DeclareFieldFormat[thesis]{institution}{#1}% suppress prefix volume of a journal
\DeclareFieldFormat[thesis]{type}{PhD Thesis}%
\DeclareFieldFormat[thesis]{year}{#1}

\DeclareBibliographyDriver{thesis}{%
  \usebibmacro{bibindex}%
  \usebibmacro{begentry}%
  \usebibmacro{author/translator+others}%
  \setunit{\addcomma\space }\newblock %separator author names and title : here it is a dot and a coma
  \setunit{\addcomma\space} % separator btw author and title
  \printfield{title}
  \setunit{\addcomma\space} % separator btw title and journaltitle
  \printfield{type}
  \setunit{\addcomma\space} % separator btw title and journaltitle
  \printlist{institution}
  \setunit{\addcomma\space} % separator btw title and journaltitle
  \setunit{\addcomma\space}\newblock  %comma after journal
  \printfield{year}
  \usebibmacro{finentry}}

\DeclareFieldFormat[book]{author}{#1}%
\DeclareFieldFormat[book]{title}{#1}% 
\DeclareBibliographyDriver{book}{%
  \usebibmacro{bibindex}%
  \usebibmacro{begentry}%
  \usebibmacro{author/translator+others}%
  \setunit{\addcomma\space }\newblock %separator author names and title : here it is a dot and a coma
  \setunit{\addcomma\space} % separator btw author and title
  \printfield{title}
  \setunit{\addcomma\space} % separator btw title and journaltitle
  \printfield{editor}
  \setunit{\addcomma\space} % separator btw title and journaltitle
  \printlist{publisher} %careful it is a list
  \setunit{\addcomma\space} % separator btw title and journaltitle
  \setunit{\addcomma\space}\newblock  %comma after journal
  \printfield{year}
  \usebibmacro{finentry}}

\usepackage{color}
\usepackage{xcolor}
\usepackage{transparent} %to make semi transparent text

\definecolor{dark_blue}{RGB}{0,9,129}
\definecolor{dark_green}{RGB}{18,172,88} 
\definecolor{dark_blue_cite}{RGB}{0,84,168}
\definecolor{light_grey}{RGB}{190,190,190}

\definecolor{RedDark}{RGB}{139,0,0}
\definecolor{Red}{RGB}{255,0,0}
\definecolor{GreenDark}{RGB}{34,139,34} %rgb '#228B22'
\definecolor{Green}{RGB}{50,205,50} %rgb '#32CD32'

\definecolor{Blueone}{RGB}{0,96,173} % HEX #0060ad
\definecolor{Redone}{RGB}{173,77,0} % HEX #ad4d00

\RequirePackage{physics}

\DeclareRobustCommand{\volfrac}[1]{%
    \IfEqCase{#1}{%
        {H2O}{\alpha_{H_{2}O}}%
        {g}{\alpha_{g}}%
        {l}{\alpha_{l}}%
        {G}{\alpha_{G}}%
        {L}{\alpha_{L}}%
        {1}{\alpha_{1}}%
        {2}{\alpha_{2}}%
        {k}{\alpha_{k}}%
        {both}{\alpha}%     
    }[\PackageError{volfrac}{Undefined option to volfrac: #1}{}]%
}%
\DeclareRobustCommand{\speed}[1]{%
    \IfEqCase{#1}{%
        {vk}{\bs{v}_{k}}%
        {sk}{v_{k}}%
        {sint}{v_{I}}%        
        {v}{\bs{v}}%
        {s}{v}%
        {s1}{v_{1}}%
        {s2}{v_{2}}%
    }[\PackageError{speed}{Undefined option to speed: #1}{}]%
}%
\newcommand{\pI}{p_{I}}     % interfaciale pressure
\newcommand{\fontscal}[1]{#1} %scalar
\newcommand{\fontscalemp}[1]{\mathsf{#1}} %wichtig scalar
\newcommand{\fontvec}[1]{\bs{\mathrm{#1}}} %vector
\newcommand{\fontmat}[1]{\bs{\mathcal{#1}}} %matrix
\newcommand{\fonttens}[1]{\bs{\mathcal{#1}}} %matrix

\newcommand{\CV}{\fontvec{u}}    % conservatif variables
\newcommand{\CF}{\fontvec{f}}    % flux conservatif
\newcommand{\NCF}{\fontmat{N}}             % flux non conservatif
\newcommand{\DF}{\fontmat{D}}             % flux dissipatif
\newcommand{\NCFone}{\fontvec{n}_{1}}     % g(2) = -g(1) = (0 , pi, uIpI)
\newcommand{\NCFtwo}{\fontvec{n}_{2}}
\newcommand{\NCFk}{\fontvec{n}_{k}}

\newcommand{\CFsym}{\fontmat{C}_{1}}       % flux conservatif
\newcommand{\CFnsym}{\fontmat{Z}_{1}}      % flux conservatif
\newcommand{\NCFsym}{\fontmat{C}_{2}}      % flux  conservatif
\newcommand{\NCFnsym}{\fontmat{Z}_{2}}     % flux non conservatif

\newcommand{\Fsymk}{\fontmat{C}_{k}}
\newcommand{\Fnsymk}{\fontmat{Z}_{k}}

\newcommand{\MES}{\fontscalemp{H}}          % mixture entropy scalar
\newcommand{\MEV}{\fontvec{v}}           % mixture entropy variable V=partial_U H
\newcommand{\MEValpha}{\mathrm{v_{\alpha}}}  % scalar component alpha of mixture entropy variable V=partial_U H
\newcommand{\PES}{s}                % phase entropy scalar

\newcommand{\ETF}{\fontscalemp{G}}       % entropy Total fluxes scalar 
\newcommand{\transmatrix}{\fontmat{T}}

\DeclareRobustCommand{\transfer}[1]{%
    \IfEqCase{#1}{%
        {salpha}{\fontscal{t_{\alpha}}} %alpha component
        {s1}{\fontscal{t_{1}}} % scalar component
        {s2}{\fontscal{t_{2}}} % scalar component
        {sk}{\fontscal{t_{k}}} % scalar component
        {v}{\fontvec{t}} % scalar component        
        {m}{\fontmat{T}}%
    }[\PackageError{transfer}{Undefined option to transfer: #1}{}]%
}%

\newcommand{\spaces}[1]{%
    \IfEqCase{#1}{%
        {M}{\mathcal{M}}%
        {TM}{\mathcal{TM}}%
        {TqM}{\mathcal{T}_{\pathM}\mathcal{M}}%
        {TR}{\mathcal{TR}}%
        {R7}{\mathbb{R}^{7}}%
        {R77}{\mathbb{R}^{7\times7}}%
        {R3}{\mathbb{R}^{3}}%
        {R5}{\mathbb{R}^{5}}%
        {R}{\mathbb{R}}%                
        {Rp}{\mathbb{R}^{p}}%                
        {Rpp}{\mathbb{R}^{p\times p}}%                
    }[\PackageError{space}{Undefined option to space: #1}{}]%
}%

\newcommand{\trans}{T} % transpose

\DeclareCiteCommand{\citejournal}
  {\usebibmacro{prenote}}
  {\usebibmacro{citeindex}%
    \usebibmacro{journal}}
  {\multicitedelim}
  {\usebibmacro{postnote}}

\newcommand{\citeay}[1]{\autocite{#1}} 
\newcommand{\specialcell}[2][c]{%
  \begin{tabular}[#1]{@{}c@{}}#2\end{tabular}}
  
\usepackage{todonotes}
   \allowdisplaybreaks
\begin{document}
%%%%%%%%%%%%%%%%%%% Publisher's Area please ignore %%%%%%%%%%%%%%%%%%%%%%%
%
%\catchline{}{}{}{}{}
%
%%%%%%%%%%%%%%%%%%%%%%%%%%%%%%%%%%%%%%%%%%%%%%%%%%%%%%%%%%%%%%%%%%%%%%%%%%
 \title{Entropy supplementary conservation law for non-linear systems of PDEs with non-conservative terms: application to the modelling and analysis of complex fluid flows using computer algebra\thanks{Received date, and accepted date (The correct dates will be entered by the editor).}}

          %For each author, make a block with the following macros:

          \author{Pierre CORDESSE\thanks{ONERA, DMPE, 8 Chemin de la Huni\`ere, 91120 Palaiseau, France, and CMAP, Ecole polytechnique, Route de Saclay 91128 Palaiseau Cedex, France, (pierre.cordesse@polytechnique.edu)}
          \and Marc MASSOT\thanks{CMAP, Ecole polytechnique, Route de Saclay, 91128 Palaiseau Cedex, France, (marc.massot@polytechnique.edu)}
          }

         \pagestyle{myheadings} \markboth{Entropy conservation law for nonlinear systems of PDEs with non-conservative terms}{CORDESSE, P. and MASSOT, M.} \maketitle
\begin{abstract}
In the present contribution, we investigate first-order nonlinear systems of partial differential equations which are constituted of two parts: a system of conservation laws and non-conservative first order terms.
Whereas the theory of first-order systems of conservation laws is well established and the conditions for the existence of supplementary conservation laws, and more specifically of an entropy supplementary conservation law for smooth solutions, well known, there exists so far no general extension to obtain  such supplementary conservation laws when non-conservative terms are present.
We propose a framework in order to extend the existing theory and show that the presence of non-conservative terms somewhat complexifies the problem since numerous combinations of the conservative and non-conservative terms can lead to a supplementary conservation law.
We then identify a restricted framework in order to design and analyze physical models of complex fluid flows by means of computer algebra and thus obtain the entire ensemble of possible combination of conservative and non-conservative terms with the objective of  obtaining specifically an entropy supplementary conservation law.
The theory as well as developed computer algebra tool are then applied to a Baer-Nunziato two-phase flow model and to a multicomponent plasma fluid model. The first one is a first-order fluid model, with non-conservative terms impacting on the linearly degenerate field and requires a closure since there is no way to derive interfacial quantities from averaging principles and we need guidance in order to close the pressure and velocity of the interface and the thermodynamics of the mixture. The second one involves first order terms for the heavy species coupled to second order terms for the electrons, the non-conservative terms impact the genuinely nonlinear fields and the model can be rigorously derived from kinetic theory. We show how the theory allows to recover the whole spectrum of closures obtained so far in the literature for the two-phase flow system as well as conditions when one aims at extending the thermodynamics and also applies to the plasma case, where we recover the usual entropy supplementary equation, thus assessing the effectiveness and scope of the proposed theory.
\end{abstract}

\begin{keywords} Nonlinear PDEs with non-conservative terms, supplementary conservation law, entropy, computer algebra, two-phase flow, Baer-Nunziato model, multicomponent plasma fluid model
\end{keywords}

\begin{AMS} 35L60; 68W30; 76N15; 76T10; 82D10
\end{AMS}
%     35L60 (Partial differential equations: hyperbolic equations and systems : nonlinear first-order hyperbolic equations)\\
%     68W30 (Symbolic computation and algebraic computation)
%     76N15 (Fluid mechanics: compressible fluids and gas dynamics, general)\\
%     76T10 (Liquid-gas two-phase flows, bubbly flows)\\
%     82D10 (Statistical mechanics, structure of matter: Plasmas)}

\section{Introduction}

First-order nonlinear systems of partial differential equations and more specifically systems of conservation laws have been the subject of a vast literature since the second half of the twentieth century because they are ubiquitous in mathematical modelling of fluid flows and are used extensively for numerical simulation in applications and industrial context \citeay{Bissuel_2018, Gaillard_2016}. %
Such systems of equation can either be rigorously derived from kinetic theory of gases through various expansion techniques \citeay{Ferziger_1972,Woods_1975}, or can be derived using rational thermodynamics and fluid mechanics including stationary action principle (SAP)  \citeay{Serrin_1959,Landau_1976,Truesdell_1969}. As far as Euler or Navier-Stokes equations are concerned for a gaseous flow field, the outcome of both approaches are similar and the mathematical properties of these systems have been thoroughly investigated for the past decades. %
%{\bf reprendre ici les propri\'et\'es en fonction de la pr\'esentation du d\'ebut}

An interesting related problem is the quest for supplementary conservation laws. Noether's theorem \citeay{Olver_1986} leads, within the framework of SAP, to the derivation of supplementary conservation laws based on symmetry transformations of the variational problem under investigation\footnote{Among the most well-known symmetry transformations, the time translation yields the conservation of the total energy of the system if the associated Lagrangian is invariant to time-shift and the space translation yields the conservation of the total momentum of the system if the Lagrangian is invariant to space-shift}. Examples of such derivations on two-phase flow modelling can be found in \citeay{Gavrilyuk_Saurel_2002, Drui_JFM_2019}. However, to the authors knowledge, no symmetry transformations have been identified yielding a conservative law on the entropy of the system. In fact, SAP does not allow to reach a closed system of equations, and one has to provide a closure for the entropy (see \citeay{Gouin_2009} for example). A specific type of supplementary conservation equation for smooth solution is especially important, namely the \emph{entropy equation}, derived through the theory developed in \citeay{Godunov_1961,Friedrichs_1971} for systems of conservation laws. Such systems of PDEs are hyperbolic at any point where a locally convex entropy function exists \citeay{Mock_1980}, and when they are equipped with a strictly convex entropy, they can be symmetrized \citeay{Friedrichs_1971} \citeay{Harten_Hyman_1983} and thus are hyperbolic. These properties have been at the heart of the mathematical theory of existence and uniqueness of smooth solutions \citeay{Kawashima_1988} \citeay{Giovangigli_1998}, but they are also a corner stone for the study of weak solutions for which the work of \citeay{Kruzkov_1970} proves the well-posedness of Cauchy problem for one-dimensional systems.

Nonetheless, for a number of applications, where reduced-order fluid models have to be used for tractable mathematical modelling and numerical simulations, be it in the industry or in other disciplines, micro-macro kinetic-theory-like approaches as well as rational thermodynamics and SAP approaches often lead to system of conservation laws involving \emph{non-conservative terms}. Among the large spectrum of applications, we focus on two types of models, which exemplify the two approaches: %
1- two phase flows models which rely on a hierarchy of diffuse interface models among which stands the Baer-Nunziato \citeay{Baer_Nunziato_1986} model used when full disequilibrium of the phases must be taken into account. Since this model is derived through rational thermodynamics, the macroscopic set of equations can not be derived from physics at small scale of interface dynamics and thus require closure of interfacial pressure and velocity, %
2- multicomponent fluid modelling of plasmas flows out of thermal equilibrium, where the equations can be derived rigorously from kinetic theory using a multi-scale Chapman-Enskog expansion mixing a hyperbolic scaling for the heavy species and a parabolic scaling for the electrons \citeay{Graille_2007}. %
Concerning the thermodynamics, whereas for the first model it has to be postulated and requires assumptions, it can be obtained from kinetic theory in the second model. In both cases, the models involve non-conservative terms, but these terms do not act on the same fields; linearly degenerate field is impacted for the two-phase flow model, whereas it acts on the genuinely nonlinear fields in the second \citeay{Wargnier_2018}. Whereas hyperbolicity depends on the closure and is not guaranteed for the first class of models \citeay{Gallouet_2004}, the second is naturally hyperbolic \citeay{Graille_2007} and also involves second-order terms and eventually source terms \citeay{Magin_2009}. 

Thus, the presence of \emph{non-conservative terms} encompasses several situations and requires a general theoretical framework. 
While Noether's theorem can still applied to obtain some supplementary conservation laws, it does not permit to exhibit all of them and especially not an entropy supplementary conservation law. A unifying theory extending the standard approach for systems of conservations laws (entropy supplementary conservation law, entropic symmetrization, Godunov-Mock theorem, hyperbolicity) is still missing for such systems even if some key advances exist. The system has been shown to be symmetrizable by \citeay{Coquel_2013} -- not in the sens of Godunov-Mock -- far from the resonance condition for which hyperbolicity degenerates. In \citeay{Forestier_2011}, the model is proved to be partially symmetrizable in the sense of Godunov-Mock.
% 

% ANNOUNCE CONTRIBUTION OF THE PAPER
The present paper first proposes an extension of the theory of supplementary conservation laws for system of conservation laws to first-order nonlinear systems of partial differential equations which are constituted of two parts: a system of conservation laws and \emph{non-conservative first order terms}.%
We emphasize how the presence of non-conservative terms somewhat complexifies the problem since numerous combinations of the conservative and non-conservative terms can lead to supplementary conservation laws. %We also provide conditions under which an entropic symmetrization of the system of equations in the sense of Godunov-Mock can be achieved. % 
%
%Furthermore, since the hyperbolic character of the system of equations in not guaranteed and the thermodynamics does not always lead to a strictly convex entropy, we need to investigate the notion of spectrum and partial symmetrization and extend the framework proposed in \citeay{Forestier_2011} to a more general class of systems. %
%
We then identify a restricted framework in order to design and analyze physical models of complex fluid flows by means of computer algebra and thus obtain the entire ensemble of possible combination of conservative and non-conservative terms to obtain an entropy supplementary conservation law. %
%\notes{Pierre: il me semble bon de pr\'eciser que notre th\'eorie est valide du d\'ebut à la fin pour n'importe quelle fonction, convexe ou pas, donc valie pour obtenir n'importe quelle loi de conservation suppl\'ementaire. Mais que nous avons d\'evelopp\'e cette th\'eorie dans le but d'obtenir celle d'entropie, pour plus tard pouvoir symm\'etriser le syst\`eme}
%
%
The proposed theoretical approach is then applied to the two systems identified so far for their diversity of behaviour. Even if the whole theory is valid for any supplementary conservation law, we focus on obtaining an \emph{entropy} supplementary conservation law. %
For the two-phase flow model, assuming a thermodynamics of non-miscible phases, we derive conditions to obtain an entropy supplementary conservative equation together with a compatible thermodynamics and closures for the non-conservative terms. Interestingly enough, all the closures proposed so far in the literature are recovered \citeay{Baer_Nunziato_1986, Kapila_1997, Bdzil_1999, Lochon_PhdThesis_2016, Saurel_Gavrilyuk_2003}. %
The strength of the formalism lies also in the capacity to derive such conditions for some level of mixing of the phases. By introducing a mixing term in the definition of the entropy, the new theory brings out constraints on the form of the added mixing term. We recover not only the closure proposed to account for a configuration energy as in the context of deflagration-to-detonation \citeay{Baer_Nunziato_1986} or in \citeay{Coquel_2002}, but we also rigorously find new closures leading to a conservative system of equations\footnote{Such closure is similar to the one used in \citeay{Powers_1988, Powers_1990} which led to a controversy \citeay{Drew_1983, Bdzil_1999, Andrianov_2003}}. %
%
%Finally, we prove that for both thermodynamics employed, the system cannot admit an entropic symmetrization and we make the link with the partial symmetrization in \citeay{Forestier_2011}. %
%
We also prove that the theory encompasses the plasma case, where we recover the usual \emph{entropy} supplementary equation assessing the effectiveness and scope of the proposed theory.

%% SPECTRUM 
%Finally, since the non-conservative systems are not always symmetrizable in the sense of Godunov-Mock, it is hard to study the hyperbolicity of these systems. In the work of \citeay{Forestier_2011} the authors have proposed for a specific class of non-conservative systems to derive conditions to guaranty the hyperbolicity of the system. To do so they separate conservative and non-conservative variables and symmetrize partially the system. The present paper proposes a generalization of the approach to study the hyperbolicity of a larger range of non-conservative systems such as the plasma model for which the methodology in \citeay{Forestier_2011} does not apply.
%
%

% ANNOUNCE OUTLINE
%{\bf reprendre l'outline - marc}
The paper is organized as follows. The extension of the theory for system of conservation laws to first-order nonlinear systems of partial differential equations 
including non-conservative terms, as well as the framework to apply the theory by means of computer algebra are introduced in Section~\ref{sec:theory}. These results are then applied first to the Baer-Nunziato model in Section~\ref{sec:BNZ} and then to the plasma model in Section~\ref{sec:plasma} to obtain an entropy supplementary conservation law compatible with the model closure.

%First and second-order non-linear systems of conservative laws have been widely investigated for the past decades starting with the work of \citeay{Godunov_1961} and then \citeay{Friedrichs_1971} \citeay{Mock_1980} \citeay{Harten_Hyman_1983} \citeay{Kawashima_1988} and write:
%\begin{align}\label{eq:second_order_non_cons_syst}
%  \partial_{t} \CV + \left\lbrace \partial_{\CV} \CF(\CV) + \NCF(\CV) \right\rbrace \partial_{x} \CV = %
%           \partial_{x} \left( \diffusion (\CV) \partial_{x} \CV \right) + \sources (\CV)
%\end{align}
%where $\CV \in \spaces{R}^{p}$ is the variable vector of conservative variables, $p$ being the number of components, $\CF \in \spaces{R}^{p}$ is a vector-valued function of $p$ components representing the conservative fluxes, $\NCF \in \spaces{R}^{p \times p}$ is a $p$-square matrix containing the non-conservative terms of the system, $\diffusion \in \spaces{R} $ is the diffusion coefficient, $\sources \in \spaces{R}^{p}$ is the source term vector.
%\newpage
\textbf{Notations:} Let $\bs{a} \in \spaces{R}^{p}$, $\bs{b} \in \spaces{R}^{p}$, $\mathcal{B} \in \spaces{R}^{p\times p}$, $\mathcal{C} \in \spaces{R}^{p\times p}$, $\fonttens{D} \in \spaces{R}^{p\times p \times p}$  be a $p$-component line first-order tensor, a $p$-component column first-order tensor, two $p$-square second-order tensor and a third-order tensor respectively. We introduce the following notations:
\begin{itemize}
  \item $\bs{a} \mathcal{B}$ is a line first-order tensor in $\spaces{Rp}$ whose $i$ component are defined by
  \begin{align}
    \left(\bs{a} \mathcal{B}\right)_{i} = \sum_{j=1,p} \bs{a}_{j}  \mathcal{B}_{j,i},
  \end{align}
  \item $\mathcal{B} \bs{b}$ is a column first-order tensor in $\spaces{Rp}$ whose $i$ component is defined by
  \begin{align}
    \left(\mathcal{B} \bs{b}\right)_{i} = \sum_{j=1,p} \mathcal{B}_{i,j}\bs{b}_{j},
  \end{align}
  \item $\mathcal{B} \times \mathcal{C}$ is $p$-square second-order tensor whose $(i,j)$ component is defined by
  \begin{align}
    \left(\mathcal{B} \times \mathcal{C}\right)_{i,j} = \sum_{k=1,p} \mathcal{B}_{i,k} \mathcal{C}_{k,j},
  \end{align}
  \item $\bs{a} \otimes \mathbb{D}$ is a $p$-square second-order tensor whose $(i,j)$ component is defined by
  \begin{align}
    \left( \bs{a} \otimes \mathbb{D} \right)_{(i,j)} = \sum_{k=1,p} \bs{a}_{k} \times \mathbb{D}_{k,i,j}.
  \end{align}  
\end{itemize}
Hereafter, we will name zero- first- and second-order tensors by scalar, vector and matrix respectively and for convenience we will use vector and matrix representations of functions. Moreover, given a scalar function $S$, the partial differentiation of $S$ by a column vector $\bs{a}$, $\partial_{\bs{a}}S$ is a line vector in $\spaces{Rp}$. Finally, $\cdot$ denotes the Euclidean scalar product in $\spaces{Rp}$.
\section{Supplementary conservation law}\label{sec:theory}

First we recall the theory of the existence of a supplementary conservative equation for first-order nonlinear systems of conservation laws. Second, this notion is extended to systems containing first order non-conservative terms. Third, we introduce a framework to apply this new theory to design and analyze physical models using computer algebra.

A one-dimensional framework is adopted from now on, $x \in \spaces{R}$, in order to simplify the derivation. Nonetheless, the results 
can easily be extended to the multi-dimensional approach as presented in  \citeay{Godlewski_1996} for systems of conservation laws.
\subsection{First-order nonlinear conservative systems.}\label{ssec:CS_theory}
The homogeneous form of a first-order nonlinear system of $p$ conservation laws writes
\begin{align}\label{eq:cons_syst_non_linear}
  \partial_{t} \CV + \partial_{x} \CF(\CV) = \fontvec{0},
\end{align}
where $\CV \in \Omega \subset \spaces{R}^{p}$ denotes the conservative variables with $\Omega$ an open convex of $\spaces{R}^{p}$ and $\CF: \CV \in \Omega \mapsto \spaces{R}^{p}$ the conservative fluxes. Focusing on smooth solution of the system~\eqref{eq:cons_syst_non_linear}, its quasi-linear form is given by
\begin{align}\label{eq:cons_syst_quasi_linear}
  \partial_{t} \CV + \partial_{\CV} \CF(\CV) \, \partial_{x} \CV = \fontvec{0}.
\end{align}
\begin{theorem} \label{theo:cons_syst_SCE}
Let $\MES: \CV \in \Omega \mapsto \spaces{R}$ be a scalar function, not necessarily convex. The following statements are equivalent:
\begin{enumerate}[label=\text{\normalfont (}\subscript{C}{\arabic*}\text{\normalfont )},noitemsep]
  \item System~\eqref{eq:cons_syst_non_linear} admits a supplementary conservative equation
  \begin{align}\label{eq:cons_syst_entropy_eq}
    \partial_{t} \MES(\CV) + \partial_{x} \ETF(\CV) = \fontvec{0},
  \end{align}
  where $\CV \in \spaces{R}^{p}$ is a smooth solution of System~\eqref{eq:cons_syst_non_linear} and $\ETF: \CV \in \Omega \mapsto \spaces{R}$ is a scalar function.
  \item There exists a scalar function $\ETF: \CV \in \Omega \mapsto \spaces{R}$ such that 
  \begin{align}\label{eq:cons_syst_compatibility_eq}
    \partial_{\CV} \MES(\CV) \, \partial_{\CV} \CF(\CV) = \partial_{\CV} \ETF(\CV).
  \end{align}
  \item $\partial_{\CV\CV} \MES(\CV) \times \partial_{\CV} \CF(\CV)$ is a $p$-square symmetric matrix.
\end{enumerate}
\end{theorem}
\begin{proof}
The proofs of the theorem can be found in the literature. We would like to recall how the last statement is obtained. Assuming $(C_{2})$, differentiating Equation~\eqref{eq:cons_syst_compatibility_eq} leads to
\begin{align}\label{eq:cons_syst_sym_cond}
  \partial_{\CV\CV} \MES(\CV) \times \partial_{\CV} \CF(\CV) + \partial_{\CV} \MES(\CV) \otimes \partial_{\CV\CV} \CF(\CV) = \partial_{\CV\CV} \ETF(\CV),
\end{align}
where $\partial_{\CV} \MES(\CV)\otimes \partial_{\CV\CV} \CF(\CV)$ is a $p$-square matrix defined as $\sum_{i} \partial_{\CV_{i}} \MES(\CV) \, \partial_{\CV \CV} \CF_{i}(\CV)$ which is a linear combination of Hessian matrices and hence symmetric. Moreover, the RHS of Equation~\eqref{eq:cons_syst_sym_cond} $\partial_{\CV\CV} \ETF(\CV)$ is symmetric. Therefore  $\partial_{\CV\CV} \MES(\CV)\times \partial_{\CV} \CF(\CV)$ is symmetric.

\end{proof}

Theorem~\ref{theo:cons_syst_SCE} applies for any type of supplementary conservative equations and other formulations of Theorem~\ref{theo:cons_syst_SCE} can be found in the literature \citeay{Harten_Hyman_1983, Godlewski_1996, Despres_2005}.
\subsection{Extension to systems involving non-conservative terms.}\label{ssec:NC_theory}
Let us now consider the homogeneous form of a first-order nonlinear system of partial differential equations constituted of two parts: conservations laws and first-order non-conservative terms. Its quasi-linear form can be written as
\begin{align}\label{eq:NC_syst}
  \partial_{t} \CV + \left[ \partial_{\CV} \CF(\CV) + \NCF(\CV) \right] \partial_{x} \CV = \fontvec{0},
\end{align}
where $\CV \in \Omega \subset \spaces{R}^{p}$ is a smooth solution with $\Omega$ an open convex of $\spaces{R}^{p}$, $\CF: \CV \in \Omega \mapsto \spaces{Rp}$ the conservative fluxes, $\NCF:\CV \in \Omega \mapsto \spaces{Rpp}$ the $p$-square matrix containing the first-order non-conservative terms.

In the following we extend the theory introduced in Section~\ref{ssec:CS_theory} to system~\eqref{eq:NC_syst}.
%
%\subsubsection{Supplementary conservative equation}
%
Given a scalar function $\MES: \CV \in \Omega \mapsto \spaces{R}$, multiplying system~\eqref{eq:NC_syst} by the line vector $\partial_{\CV} \MES(\CV)$ yields
\begin{align}\label{eq:NC_syst_SCE}
  \partial_{t} \MES + \partial_{\CV} \MES(\CV)  \left[ \partial_{\CV} \CF(\CV) + \NCF(\CV) \right] \partial_{x}\CV = 0.
\end{align}
Compared to Equation~\eqref{eq:cons_syst_entropy_eq}, the presence of the non-conservative terms in Equation~\eqref{eq:NC_syst_SCE} complexifies the question of the existence of a supplementary conservative equation. Therefore we propose to decompose in a specific way the conservative and non-conservative terms in Definition~\ref{def:decomposition}.
\begin{definition}\label{def:decomposition} Given a scalar function $\MES: \CV \in \Omega \mapsto \spaces{R}$ and a first-order nonlinear non-conservative system~\eqref{eq:NC_syst}, let us define the four $p$-square matrices, $\CFsym(\CV)$, $\CFnsym(\CV)$, $\NCFsym(\CV)$ and $\NCFnsym(\CV)$ in $\spaces{R}^{p \times p}$ such that
\begin{align}
  \partial_{\CV} \CF(\CV) & = \CFsym(\CV) + \CFnsym(\CV), \\
  \NCF(\CV) & = \NCFsym(\CV)+\NCFnsym(\CV),
\end{align}
with the condition
\begin{align}\label{eq:decomposition_condition}
  \partial_{\CV} \MES(\CV) \left[ \CFnsym(\CV) + \NCFnsym(\CV) \right] = \fontvec{0}.
\end{align}
\end{definition}
%\begin{remark}
%Since there are quantities of combinations of the conservative and non-conservative terms, the definition of the matrices $\Fsymk$ and $\Fnsymk$ that fulfils Definition~\ref{def:decomposition} is not unique.
%\end{remark}
%
In light of Definition~\ref{def:decomposition}, Theorem~\ref{theo:cons_syst_SCE} can be extended as follows:
\begin{theorem}\label{theo:NC_syst_SCE}
Let $\MES: \CV \in \Omega \mapsto \spaces{R}$ be a scalar function, not necessarily convex. Given a first-order nonlinear system of non-conservation laws \eqref{eq:NC_syst}, if we introduce the decomposition as in Definition~\ref{def:decomposition}, then the following statements are equivalent:
\begin{enumerate}[label=(\subscript{C}{\arabic*}),noitemsep]
  \item System~\eqref{eq:NC_syst} admits a supplementary conservative equation
  \begin{align}\label{eq:NC_syst_entropy_eq}
    \partial_{t} \MES(\CV) + \partial_{x} \ETF(\CV) = 0,
  \end{align}
  where $\CV \in \spaces{R}^{p}$ is a smooth solution of System~\eqref{eq:NC_syst} and $\ETF: \CV \in \Omega \mapsto \spaces{R}$ is a scalar function.
  \item There exists a scalar function $\ETF: \CV \in \Omega \mapsto \spaces{R}$ such that 
  \begin{align}\label{eq:NC_syst_compatibility_eq}
    \partial_{\CV} \MES(\CV) \left[ \CFsym(\CV) + \NCFsym(\CV) \right]  &= \partial_{\CV} \ETF(\CV).
  \end{align}
  \item $\partial_{\CV\CV} \MES(\CV)\times \left[ \CFsym(\CV) + \NCFsym(\CV) \right]+ \partial_{\CV} \MES(\CV)\otimes \partial_{\CV} \left[\CFsym(\CV) + \NCFsym(\CV) \right] $ is a $p$-square symmetric matrix.
\end{enumerate}
\end{theorem}
\begin{proof}
Rewriting Equation~\eqref{eq:NC_syst_SCE} using the decomposition of the conservative and non-conservative terms as
\begin{align}
  \partial_{t} \MES(\CV) + \partial_{\CV} \MES(\CV) \left[ \CFsym(\CV) + \NCFsym(\CV) \right] \partial_{x}\CV = - \partial_{\CV} \MES(\CV) \left[ \CFnsym(\CV) + \NCFnsym(\CV) \right] \partial_{x}\CV
\end{align}
outlines the result.
\end{proof}
\begin{remark}\label{remark:symmetry_condition}Theorem~\ref{theo:NC_syst_SCE} applies for any type of supplementary conservative equations. The usual symmetry condition on which relies the existence of a supplementary conservation equation is strongly modified when non-conservation terms are present. From Theorem~\ref{theo:cons_syst_SCE} to Theorem~\ref{theo:NC_syst_SCE} the condition
\vspace{-1em}
\begin{align*}
  \partial_{\CV\CV} \MES(\CV) \times \partial_{\CV} \CF(\CV) \text{ symmetric},
\end{align*}
is modified into
\vspace{-1em}
\begin{align*}
  \partial_{\CV\CV} \MES(\CV)\times \left[ \CFsym(\CV) + \NCFsym(\CV) \right]+ \partial_{\CV} \MES(\CV)\otimes \partial_{\CV} \left[\CFsym(\CV) + \NCFsym(\CV) \right]  \text{ symmetric}.
\end{align*}
%for the former. 
%One notices the supplementary term function of the first-order partial derivative of $\MES(\CV)$, $\CFsym(\CV)$ and $\NCFsym(\CV)$.
%
In the context of systems of conservation laws, an interesting algebraic approach is proposed in \citeay{Barros_2006} based on the reinterpretation of the symmetric Condition $(C_{3})$ in Theorem~\ref{theo:cons_syst_SCE}  as a Frobenuis problem. Nevertheless, when dealing with additional non-conservative terms, the above new symmetry condition prevents us from 
applying efficiently such an approach.
\end{remark}
\begin{remark} In Definition~\ref{def:decomposition}, the condition~\eqref{eq:decomposition_condition} implies that the conservative and non-conservative terms depend only on the variables $\CV$, and not on their gradient. Some authors have allowed the matrices $\Fnsymk$ to depend also on the gradients of the variables $\CV$, then a more general condition for the decomposition can be written
\begin{align}\label{eq:theo_b1b2_extended}
  \partial_{\CV} \MES(\CV) \left[ \CFnsym(\CV, \partial_{x}\CV) + \NCFnsym(\CV, \partial_{x}\CV) \right] \partial_{x}\CV  \leq 0.
\end{align}
In Section~\ref{sec:BNZ}, we will see that such a condition has been chosen to close the Baer-Nunziato model \citeay{Saurel_Gavrilyuk_2003}. However, since it changes the mathematical nature of the PDE under investigation, we will not include it in our study.
\end{remark}

From a modelling perspective, System~\eqref{eq:NC_syst} under consideration is not necessary closed. Therefore, the following corollary yields conditions on the model to obtain a supplementary conservative equation once we have postulated the thermodynamics.
\begin{corollary}
Let $\MES: \CV \in \Omega \mapsto \spaces{R}$ be a scalar function, not necessarily convex. Given a first-order nonlinear system of non-conservation laws \eqref{eq:NC_syst} where $\CF: \CV \in \Omega \mapsto \spaces{Rp}$ and $\NCF:\CV \in \Omega \mapsto \spaces{Rpp}$ are unknown functions to be modelled. If we introduce the decomposition as in Definition~\ref{def:decomposition}, then System~\eqref{eq:NC_syst} admits a supplementary conservative equation
  \begin{align}\label{eq:NC_syst_entropy_eq_corro}
    \partial_{t} \MES(\CV) + \partial_{x} \ETF(\CV) = 0,
  \end{align}
  where $\CV \in \Omega \subset \spaces{R}^{p}$ is a smooth solution of System~\eqref{eq:NC_syst} and $\ETF: \CV \in \Omega \mapsto \spaces{R}$ a scalar function, if and only if the following conditions hold
\begin{enumerate}[label=(\subscript{C}{\arabic*}),noitemsep]
  \item $\partial_{\CV\CV} \MES(\CV)\times \left[ \CFsym(\CV) + \NCFsym(\CV) \right]+ \partial_{\CV} \MES(\CV)\otimes \partial_{\CV} \left[\CFsym(\CV) + \NCFsym(\CV) \right] $ is a $p$-square symmetric matrix.
  \item $\partial_{\CV} \MES(\CV) \left[ \CFnsym(\CV) + \NCFnsym(\CV) \right] = \fontvec{0}$.
\end{enumerate}
\end{corollary}

\subsection{Design or analysis of physical models using computer algebra.}\label{ssec:NC_theory_applied}
We would like to apply the theory on first-order nonlinear non-conservative systems introduced in Section~\ref{ssec:NC_theory} to physical models such as the Baer-Nunziato model and the plasma model in order to design and analyze them. We recall that our prior interest is to obtain an \emph{entropy} supplementary conservation law. However, the difficulty is manifold:
\begin{itemize}[label=$-$]
  \item The combination of the non-conservative terms and conservative terms proposed in Definition~\ref{def:decomposition} to build a supplementary conservative equation is not unique and thus many degrees of freedom exist in defining the matrices $\Fsymk$ and $\Fnsymk$.
  \item When the model is derived trough rational thermodynamics, terms in the system of equations might need closure and the thermodynamics has to be postulated. Therefore, the matrices $\Fsymk$ and $\Fnsymk$ can contain unknowns related to the system and the definition of $\MES$.
  \item The calculations needed to derive a supplementary conservative equation are heavy and choice-based. Any change of $\Fsymk$ and $\Fnsymk$ that respects Definition~\ref{def:decomposition}, or any new postulated thermodynamics would require to derive again all the equations, and eventually a very limited range of possibilities would be examined.
\end{itemize}
These difficulties to apply the theory and examine all the possibilities makes computer algebra very appealing since it allows symbolic operations to be implemented and thus can derive equations systematically and quasi-instantaneously for any combinations of conservative and non-conservative terms as well as model closure and $\MES$ definition. 

Furthermore, the generic level handled by computer algebra is not unlimited and therefore Definition~\ref{def:decomposition} requires further assumptions to circumscribe the number of degrees of freedom that can be accounted for.

Even if the theory proposed hereinbefore is valid to obtain any kind of supplementary conservation laws, we are mainly interested 
in obtaining an entropy supplementary conservation law. We thus need  to define the notions of \textit{entropy} and \textit{entropic variables} in the following two definitions.

\begin{definition} $\MES: \CV \in \Omega \mapsto \spaces{R}$ is said to be an \textit{entropy} of the system~\eqref{eq:NC_syst} if $\MES(\CV)$ is a convex scalar function of the variables $\CV$ which fulfills Theorem~\ref{theo:cons_syst_SCE}. The supplementary conservative equation~\eqref{eq:cons_syst_entropy_eq} is then named the \textit{entropy equation} and $\ETF: \CV \in \Omega \mapsto \spaces{R}$ is the associated \textit{entropy flux}.
\end{definition}

\begin{definition}\label{def:entropic_variable}
Let $\MES: \CV \in \Omega \mapsto \spaces{R}$ be a scalar function, not necessarily convex. Given a first-order nonlinear conservative system~\eqref{eq:cons_syst_non_linear}, let us define the \textit{entropic variables} $\MEV: \CV \in \Omega \mapsto \spaces{Rp}$ such that
\begin{align}
  \MEV(\CV) = \left( \partial_{\CV} \MES(\CV)\right)^{t}.
\end{align}
\end{definition}

The entropic variables have been studied in \citeay{Giovangigli_1998} in order to obtain symmetric and normal forms of the system of equation and used in the framework of gaseous mixtures, where 
the mathematical entropy $\MES$ is usually defined as the opposite of a physical entropy density per unit volume of the system \citeay{Giovangigli_1998}.

\begin{definition}\label{def:decomposition_applied} Given a scalar function $\MES: \CV \in \Omega \mapsto \spaces{R}$, a first-order nonlinear non-conservative system~\eqref{eq:NC_syst}, and the four $p$-square matrices $\CFsym(\CV)$, $\CFnsym(\CV)$, $\NCFsym(\CV)$ and $\NCFnsym(\CV)$ in $\spaces{Rpp}$ defined in Definition~\ref{def:decomposition}, we introduce the unknown line vector $\transfer{v}: \CV \in \Omega \mapsto \spaces{Rp}$ such that
\begin{align}
    \partial_{\CV} \MES(\CV) \left[ \CFsym(\CV) + \NCFsym(\CV) \right] &= \partial_{\CV} \MES(\CV) \, \partial_{\CV} \CF(\CV) + \transfer{v}(\CV),\\
  \partial_{\CV} \MES(\CV) \left[ \CFnsym(\CV) + \NCFnsym(\CV) \right] &= \partial_{\CV} \MES(\CV) \, \NCF(\CV) - \transfer{v}(\CV).  
\end{align}
The condition of Equation~\eqref{eq:decomposition_condition} rewrites into
\begin{align}\label{eq:decomposition_condition_applied}
  \partial_{\CV} \MES(\CV) \, \NCF(\CV) - \transfer{v}(\CV) = \fontvec{0}.
\end{align}
\end{definition}
\begin{remark}
Since Definition~\ref{def:decomposition_applied} is a projection of the matrix equations of Definition~\ref{def:decomposition} on the vector $\partial_{\CV} \MES(\CV)$, it may be interesting to introduce an unknown matrix $\transmatrix(\CV) \in \spaces{R}^{p \times p}$ associated to the unknown line vector $\transfer{v}(\CV)$ such that
\begin{align}\label{def:gamma_matrix}
  \transfer{v}(\CV) = \partial_{\CV} \MES(\CV) \transmatrix (\CV).
\end{align}
Thus, Definition~\ref{def:decomposition_applied} can be formulated as follows
\begin{align}
  \CFsym(\CV) + \NCFsym(\CV) &= \partial_{\CV} \CF(\CV) + \transmatrix(\CV), \\
  \CFnsym(\CV) + \NCFnsym(\CV) &= \NCF(\CV) - \transmatrix(\CV),
\end{align}
with the condition
\begin{align}
  \partial_{\CV} \MES(\CV) \left[ \NCF(\CV) - \transmatrix(\CV) \right] = \fontvec{0}.
\end{align}
\end{remark}

The unknown functional line vector $\transfer{v}(\CV) \in \spaces{R}^{p}$ represents the transfer of non-conservative terms to the conservative terms. In the degenerate case where $\transfer{v}=\fontvec{0}$, $\Fsymk$ receives all the conservative terms and $\Fnsymk$ all the non-conservative terms. Condition~\eqref{eq:decomposition_condition_applied} forces all the non-conservative terms to vanish and System~\eqref{eq:NC_syst} is fully conservative, hence the theory of conservative system can be applied.

Definition~\ref{def:decomposition_applied} being more restrictive than Definition~\ref{def:decomposition}, computer algebra is now applicable to analyze the properties of a first-order nonlinear non-conservative system leading to a reformulation of Theorem~\ref{theo:NC_syst_SCE}.
\begin{theorem}\label{theo:NC_SCE_applied}
Let $\MES: \CV \in \Omega \mapsto \spaces{R}$ be a scalar function, not necessarily convex. Consider a first-order nonlinear system of non-conservation laws \eqref{eq:NC_syst}. If we introduce the decomposition as in Definition~\ref{def:decomposition_applied}, then the following statements are equivalent:
\begin{enumerate}[label=\text{\normalfont (}\subscript{C}{\arabic*}\text{\normalfont )},noitemsep]
  \item System~\eqref{eq:NC_syst} admits a supplementary conservative equation
  \begin{align}\label{eq:entropy_eq}
    \partial_{t} \MES(\CV) + \partial_{x} \ETF(\CV) = 0,
  \end{align}
  where $\CV \in \spaces{R}^{p}$ is a smooth solution of System~\eqref{eq:NC_syst} and $\ETF: \CV \in \Omega \mapsto \spaces{R}$ is a scalar function.
  \item There exists a scalar function $\ETF: \CV \in \Omega \mapsto \spaces{R}$ such that 
  \begin{align}\label{eq:compatibility_eq}
    \partial_{\CV} \MES(\CV) \, \partial_{\CV} \CF(\CV) + \transfer{v}(\CV) &= \partial_{\CV} \ETF(\CV).
  \end{align}
  \item $\partial_{\CV \CV} \MES(\CV) \times \partial_{\CV} \CF(\CV) +  \partial_{\CV} \transfer{v}(\CV)$ is a $p$-square symmetric matrix.
\end{enumerate}
\end{theorem}
\begin{proof}
Injecting Definition~\ref{def:decomposition_applied} into Theorem~\ref{theo:NC_syst_SCE} leads to these results.
\end{proof}

When $\MES$ is the entropy of the system, Theorem~\ref{theo:NC_SCE_applied} provides equations that relate the thermodynamics of the model through $\MES$, the model itself with possible terms to be closed in $\CF(\CV)$ and $\NCF(\CV)$, and the unknown line vector $\transfer{v}(\CV)$. Combined with the Definition~\ref{def:decomposition_applied}, Theorem~\ref{theo:NC_SCE_applied} brings out conditions on the model to obtain a supplementary conservative equation given a postulated thermodynamics and it leads to the following corollary.
\begin{corollary}\label{coro:NC_syst_applied_metho}
Consider a first-order nonlinear system of non-conservation laws \eqref{eq:NC_syst} where $\CV \in \Omega \subset \spaces{R}^{p}$ is a smooth solution with $\Omega$ an open convex of $\spaces{R}^{p}$ but $\CF: \CV \in \Omega \mapsto \spaces{Rp}$ and $\NCF:\CV \in \Omega \mapsto \spaces{Rpp}$ are unknown functions to be modelled. 
Let $\MES: \CV \in \Omega \mapsto \spaces{R}$ be a scalar function, not necessarily convex. If we introduce the decomposition as in Definition~\ref{def:decomposition_applied}, then System~\eqref{eq:NC_syst} admits a supplementary conservative equation
  \begin{align}\label{eq:NC_syst_entropy_eq_corro_applied}
    \partial_{t} \MES(\CV) + \partial_{x} \ETF(\CV) = 0,
  \end{align}
  where $\ETF: \CV \in \Omega \mapsto \spaces{R}$ is a scalar function
if and only if the following conditions hold
\begin{enumerate}[label=\text{\normalfont (}\subscript{C}{\arabic*}\text{\normalfont )},noitemsep]
  \item $\partial_{\CV \CV} \MES(\CV) \times \partial_{\CV} \CF(\CV)  +  \partial_{\CV} \transfer{v}(\CV)$ is symmetric.
  \item $\partial_{\CV} \MES(\CV) \NCF(\CV) - \transfer{v}(\CV)  = \fontvec{0}$.
\end{enumerate}
\end{corollary}

\begin{remark}
The previous framework can be extended to the multi-dimensional case in a straightforward manner. If the original system is isotropic, such as for the applications we have in mind, then the previous conditions will be the same in the various directions. In the framework of more general non-isotropic systems, which satisfy Galilean and rotational invariances for example, we will obtain different conditions and we have to check that the decomposition we perform in the various directions satisfies some compatibility relations so that the obtained conservation law satisfies the original invariance properties of the system.
\end{remark}

\subsection{Methodology.}\label{ssec:NC_metho}
Corollary~\ref{coro:NC_syst_applied_metho} draws the methodology we have implemented in the Maple{\texttrademark} computer algebra software\footnote{Maple is a trademark of Waterloo Maple Inc.} in order to obtain an \emph{entropy} supplementary conservation law. Our methodology is the following:
\begin{enumerate}[label=$(Step \ \arabic*)$]
  \item We define the thermodynamics by postulating - if need be - an entropy function $\MES: \CV \in \Omega \mapsto \spaces{R}$.
  \item We then use Condition $(C_{1})$ and $(C_{2})$ of Corollary~\ref{coro:NC_syst_applied_metho} to ensure the existence of an entropy flux $\ETF: \CV \in \Omega \mapsto \spaces{R}$ and solve
  \begin{align}\label{eq:NC_SCE_applied_cond}
\left\{ \,
\begin{IEEEeqnarraybox}[\IEEEeqnarraystrutmode
\IEEEeqnarraystrutsizeadd{1pt}{1pt}][c]{l}
\partial_{\CV \CV} \MES(\CV) \times \partial_{\CV} \CF(\CV)  +  \partial_{\CV} \transfer{v}(\CV) \text{ symmetric}, \\
\partial_{\CV} \MES(\CV) \, \NCF(\CV) - \transfer{v}(\CV)  = \fontvec{0}.
\end{IEEEeqnarraybox}
\right.
  \end{align}
  In System~\eqref{eq:NC_SCE_applied_cond}, $\transfer{v}(\CV)$ is systematically an unknown, $\CF(\CV)$, $\NCF(\CV)$ as well as $\MES(\CV)$ can include unknown terms for which the variable dependency is specified. Maple{\texttrademark} generates then an exhaustive solution for $\transfer{v}(\CV)$ and constraints on all the other unknown terms.
  \item From that, the software derives the admissible entropy flux $\ETF: \CV \in \Omega \mapsto \spaces{R}$ which gives then the supplementary conservative equation.
\end{enumerate}
%
%
%In Section~\ref{sec:BNZ}, we will see how the methodology introduced here by means of computer algebra is a powerful tool to design the Baer-Nunziato model obtained from rational thermodynamics. In section~\ref{sec:plasma}, the methodology encompasses the plasma system derived rigorously from kinetic theory and permits to analyse it.
%\newpage
\section{Application to the Baer-Nunziato model}\label{sec:BNZ}

\subsection{Context and presentation of the model.}
The Baer-Nunziato model has been derived through rational thermodynamics in \citeay{Baer_Nunziato_1986} and describes a two-phase flow out of equilibrium. Extended by the work of \citeay{Saurel_1999} thanks to the introduction of interfacial quantities, the homogeneous form of the Baer-Nunziato model is
\renewcommand\arraystretch{1}
\setlength\arraycolsep{8pt}
\begin{align}\label{sys:BNZ_eq}
\begin{IEEEeqnarraybox}[\IEEEeqnarraystrutmode
\IEEEeqnarraystrutsizeadd{1pt}{1pt}][c]{c}
  \partial_{t} \CV + \left[ \partial_{\CV} \CF(\CV) + \NCF(\CV) \right] \partial_{x} \CV = \fontvec{0},\\ \\
  \partial_{\CV} \CF(\CV) = \begin{pmatrix}
        0 & \fontvec{0} & \fontvec{0} \\
        \fontvec{0}   & \partial_{\CV_{2}} \CF_{2}(\CV_{2}) & \fontvec{0}  \\
      \fontvec{0} & \fontvec{0} & \partial_{\CV_{1}} \CF_{1}(\CV_{1})
    \end{pmatrix}, \ %
\NCF(\CV) = \begin{pmatrix}
        \speed{sint} & \fontvec{0} & \fontvec{0}\\
        \NCFtwo & \fontvec{0} & \fontvec{0}\\
      \NCFone & \fontvec{0} & \fontvec{0}
    \end{pmatrix},
\end{IEEEeqnarraybox}
\end{align}
where the column vector $\CV \in \spaces{R}^{7}$ is defined by $\CV^{\trans} = \left(\alpha_{2},\, \CV_{2}^{\trans},\, \CV_{1}^{\trans} \right)$, $\CV_{k}^{\trans} = ( \alpha_{k} \rho_{k},\, \allowbreak \allowbreak \alpha_{k} \rho_{k} \speed{sk},\, \allowbreak \alpha_{k} \rho_{k} E_{k} )$. The conservative flux $\CF: \CV \in \Omega \mapsto \spaces{R7}$ reads $\CF(\CV)^{\trans} = (0,\, \CF_{2}(\CV_{2})^{\trans},\, \CF_{1}(\CV_{1})^{\trans})$ with $\CF_{k}(\CV_{k})^{\trans} = (\alpha_{k} \rho_{k} \speed{sk},\, \allowbreak \alpha_{k}(\rho_{k} \speed{sk}^{2}+p_{k}),\, \allowbreak \alpha_{k} ( \rho_{k} E_{k}+p_{k})\speed{sk} )$. $\NCF: \CV \in \Omega \mapsto \spaces{R77}$ is the matrix containing the non-conservative terms with $\NCFtwo(\CV)^{\trans} = - \NCFone(\CV)^{\trans} = (0,\, \allowbreak  -\pI,\, \allowbreak -\pI \speed{sint})$. Then, $\alpha_{k}$ is the volume fraction of phase $k \in \left[ 1,2 \right]$, $\rho_{k}$ the partial density, $\speed{sk}$ the phase velocity, $p_{k}$ the phase pressure, $E_{k}=\epsilon_{k} + \allowbreak \speed{sk}^{2}/2$ the total energy per unit of mass, $\epsilon_{k}$ the internal energy, $\speed{sint}$ the interfacial velocity and $\pI$ the interfacial pressure.

Two levels of ingredients are still missing for this model. First, the macroscopic set of equations includes the interface dynamics through the interfacial terms $\speed{sint}$ and $\pI$ and thus needs closure on these terms. Second the thermodynamics has to be postulated.

The mathematical properties of the model have been studied by \citeay{Embid_Baer_1992, Coquel_2002, Gallouet_2004} among others and many closure have been proposed for the interfacial terms based on wave-type considerations and the entropy inequality.

Regarding the thermodynamics, for non-miscible phases, the entropy $\MES(\CV)$ is commonly defined by Equation~\eqref{eq:mixture_entropy_immiscible} as in \citeay{Coquel_2002, Lochon_PhdThesis_2016},
\begin{align}\label{eq:mixture_entropy_immiscible}
  \MES(\CV) = - \sum_{k=1,2} \alpha_{k} \rho_{k} s_{k},
\end{align}
with $s_{k}=s_{k}(\rho_{k},p_{k})$ the phase entropy which takes for the Ideal Gas equation of state the form
\begin{align}
  s_{k}= c_{v,k} \text{ln}\left(\frac{p_{k}}{\rho_{k}^{\gamma_{k}}}\right),
\end{align}
with $c_{v,k}$ the heat capacity, $p_{k}$ the pressure, $\rho_{k}$ the density and $\gamma_{k}$ the isentropic coefficient of phase $k$.

If we were to account for partial miscibility between the two phases, we would have to add a mixing term to the definition of the non-miscible entropy. The mixing term could take the form proposed in \citeay{Gallouet_2004}, so that the entropy rewrites
\begin{align}\label{eq:mixture_entropy_extended}
  \MES =  - \sum_{k=1,2} \alpha_{k} \rho_{k} \left[ \PES_{k}(\rho_{k},p_{k}) - \psi_{k}(\alpha_{k}) \right],
\end{align}
with $\psi_{k}$, $k=\left[ 1,2 \right]$, two strictly convex nonlinear arbitrary functions depending on the volume fraction. Nevertheless, so far in the literature, no explicit expressions of these functions have been proposed. In \citeay{Gallouet_2004}, in order to obtain a supplementary conservative equation using the entropy defined in Equation~\eqref{eq:mixture_entropy_extended}, the authors show that the following condition has to be fulfilled
\begin{align}\label{eq:mixing_term_condition_BNZ}
  \psi_{k}(\alpha_{k}) = \psi_{k^{\prime}}(\alpha_{k^{\prime}}).
\end{align}

In this section, we apply to the Baer-Nunziato model the framework introduced in Section~\ref{sec:theory} by means of computer algebra. We will firstly assume the phases are non-miscible and derive an entropy supplementary conservative equation along with conditions on the interfacial terms. All the closures proposed in the literature will be recovered. Secondly, we will also apply the methodology in the case of a thermodynamics with partial miscibility and derive an entropy supplementary conservative equation together with conditions on both the interfacial terms and the mixing terms of the entropy. Not only all the closures proposed in the literature are recovered but also new ones and we also propose explicit formulations of the mixing terms and show that depending on their expression, the condition expressed in \citeay{Gallouet_2004} is not necessary.

\subsection{Methodology and decomposition.}
We start without any condition on $(\speed{sint}, \pI)$. We need initially to fix a decomposition of $\partial_{\CV} \CF(\CV)$ and $\NCF(\CV)$ including a certain degree of freedom as explained in Section~\ref{ssec:NC_theory_applied}.

Given an entropy $\MES: \CV \in \Omega \mapsto \spaces{R}$ of System~\eqref{sys:BNZ_eq}, by expressing the entropic variables as $\MEV(\CV)^{\trans} = \left( \MEValpha, \MEV_{2}^{\trans}, \MEV_{1}^{\trans}\right)$, we use the decomposition proposed in Definition~\eqref{def:decomposition_applied}. Since we do not want to generate other non-conservative terms, we choose to define the line vector $\transfer{v}: \CV \in \Omega \mapsto \spaces{Rp}$ by $\transfer{v}(\CV) = \left( \transfer{salpha}(\CV), \fontvec{0}, \fontvec{0} \right)$ where $\transfer{salpha}: \CV \in \Omega \mapsto \spaces{R}$ is the unknown scalar function a priori of all the variables $\CV$. We obtain the following decompositions
\begin{IEEEeqnarray}{rClrClrClrCl}\label{eq:final_entropic_condition_decomp_BNZ}
\IEEEyesnumber\IEEEyessubnumber*
\left( \partial_{\CV} \MES \left[ \CFsym + \NCFsym \right] \right)^{\trans} &=&\begin{pmatrix}
    \transfer{salpha}(\CV)\\
    \MEV_{2} \cdot \partial_{\CV_{2}}\CF_{2}(\CV_{2})\\
    \MEV_{1} \cdot \partial_{\CV_{1}}\CF_{1}(\CV_{1})
\end{pmatrix}, \\  %
\IEEEnonumber\IEEEyessubnumber*
\left(\partial_{\CV} \MES \left[ \CFnsym + \NCFnsym \right] \right)^{\trans} &=&\begin{pmatrix}
    -\transfer{salpha}(\CV) + \MEValpha \speed{sint} + \sum\limits_{\text{\tiny $k{=}1,2$}}\MEV_{k} \cdot \NCFk\\
    \fontvec{0}\\
    \fontvec{0}
\end{pmatrix}.
\end{IEEEeqnarray}
\normalsize
$\transfer{salpha}$ allows fractions of the non-conservative terms to feed the matrix $\Fsymk$.

Given this decomposition, we use the methodology proposed in Section~\ref{ssec:NC_metho}. ($Step$ 2) will be split here into two sub-steps.
 \begin{enumerate}[label=$(Step \ 2.\alph*)$]
 \setcounter{enumi}{0}
    \item Condition $(C_{1})$ on the symmetry of the matrix $\partial_{\CV \CV} \MES(\CV) \times \partial_{\CV} \CF(\CV)  +  \partial_{\CV} \transfer{v}(\CV)$ ensures the existence of an entropy flux $\ETF(\CV)$. It will determine $\transfer{v}(\CV)$.
    \item Knowing $\transfer{v}(\CV)$, Condition $(C_{2})$, $\partial_{\CV} \MES(\CV) \, \NCF(\CV) - \transfer{v}(\CV)  = \fontvec{0}$, will return an equation linking $(\speed{sint},\pI)$ and also $\psi_{k}$ when miscibility is accounted for.
\end{enumerate}
\subsection{Non-miscible phases entropy.}
We start applying our method $(Step \ 1)$ by postulating $\MES$ as in Equation~\eqref{eq:mixture_entropy_immiscible}. The thermodynamics is entirely known and we use the Ideal Gas EOS. The entropic variables $\MEV$ are then
\begin{align}
  \MEV = \begin{pmatrix}
      \MEValpha \\
      \MEV_{2} \\
      \MEV_{1}
    \end{pmatrix} %
  \text{ with } %
  \MEValpha = \frac{p_{1}}{T_{1}} - \frac{p_{2}}{T_{2}} %
  \text{ and } %
  \MEV_{k} = \frac{1}{T_{k}}\begin{pmatrix} 
      g_{k} - \speed{sk}^{2}/2 \\
      \speed{sk} \\
      -1 
    \end{pmatrix},
\end{align}
with $g_{k}$ the Gibbs free energy, $g_{k} =\epsilon_{k} + p_{k}/\rho_{k} - T_{k}s_{k}$. We now apply the conditions to determine $\transfer{salpha}(\CV)$ and derive the equation that links the interfacial quantities $\speed{sint}$ and $\pI$.
\begin{theorem}\label{theo:BNZ_ab_classic}
Consider System~\eqref{sys:BNZ_eq}. If the mixture entropy is defined as $\MES = - \sum_{k=1,2} \alpha_{k} \rho_{k} \PES_{k}$ then with the decomposition proposed in Equations~\eqref{eq:final_entropic_condition_decomp_BNZ}
\begin{align}\label{eq:classic_cond_a}
  \partial_{\CV \CV} \MES(\CV) \times \partial_{\CV} \CF(\CV)  +  \partial_{\CV} \transfer{v}(\CV) \text{ symmetric } \Leftrightarrow \transfer{salpha} (\CV) = \fontscal{F}(\alpha_{2}) + \frac{p_{1}}{T_{1}}u_{1} - \frac{p_{2}}{T_{2}}u_{2},
\end{align}
with $\fontscal{F}$ a strictly convex arbitrary function depending on the volume fraction $\alpha_{2}$. As a consequence the condition on $\partial_{\CV} \MES(\CV) \left[ \CFnsym(\CV) + \NCFnsym(\CV) \right]$ gives
\begin{align}\label{eq:classic_cond_b}
\begin{IEEEeqnarraybox}[\IEEEeqnarraystrutmode
\IEEEeqnarraystrutsizeadd{1pt}{1pt}][c]{rl}
   \partial_{\CV} \MES(\CV) \left[ \CFnsym(\CV) + \NCFnsym(\CV) \right] &= \fontvec{0} \\
  \Leftrightarrow \ - \fontscal{F}(\alpha_{2}) + \sum_{k=1,2} \frac{(-1)^{k}}{T_{k}} (\pI - p_{k} )( \speed{sk} - \speed{sint}) &= 0.
\end{IEEEeqnarraybox}
\end{align}
\end{theorem}
\begin{proof}
The function $\transfer{salpha}$ is found relying on symbolic computation and it holds as a proof.
\end{proof}

As explained in $(Step\ 2.a)$, Equation~\eqref{eq:classic_cond_a} guarantees the existence of an entropy flux $\ETF$ associated with the mixture entropy $\MES$ chosen as in Equation~\eqref{eq:mixture_entropy_immiscible} by defining the unknown function $\transfer{salpha}(\CV)$.

Then as described in $(Step\ 2.b)$, Equation~\eqref{eq:classic_cond_b} relates the interfacial terms $(\speed{sint}, \pI)$. By choosing $\fontscal{F}(\alpha_{2}) =0$, the condition on $\partial_{\CV} \MES \times \left[ \CFnsym + \NCFnsym \right]$ writes
\begin{align}\label{eq:classic_cond_b_reduced}
\sum_{k = 1,2} \frac{1}{T_{k}} \left(p_{k}-\pI \right) \left(\speed{sint}-\speed{sk}\right) = 0.
\end{align}
So now, to obtain a closed model along with a supplementary conservative equation, we can postulate an interfacial velocity $\speed{sint}$ and derive the corresponding $\pI$. We will limit ourselves to defining $\speed{sint}$ such that the field associated to $\speed{sint}$ is linearly degenerate. In that case, the only admissible interfacial velocities are $\speed{sint} = \beta u_{1} + (1-\beta) u_{2}$ with $\beta \in \left[ 0,1, \alpha_{1} \rho_{1}/\rho \right]$ \citeay{Coquel_2002}, \citeay{Lochon_PhdThesis_2016}. We will focus on the particular case where $\fontscal{F}(\alpha_{2}) =0$. We obtain the following results:
\begin{itemize}[label=$-$]
  \item If $\speed{sint} = \speed{sk}$, then Equation~\eqref{eq:classic_cond_b_reduced} returns $\pI = p_{k^{\prime}}$. $(\speed{sk},p_{k^{\prime}})$ is the closure proposed first by \citeay{Baer_Nunziato_1986}, \citeay{Kapila_1997}, \citeay{Bdzil_1999}, in the context of deflagration-to-detonation.
  \item If $\speed{sint} = \beta u_{1} + (1-\beta) u_{2}$ with $\beta = \alpha_{1} \rho_{1}/\rho$, then Equation~\eqref{eq:classic_cond_b_reduced} returns $\pI =\mu p_{1} + (1- \mu ) p_{2}$ with $\mu\left(\beta\right) =(1-\beta) T_{2} / ( \beta T_{1} + (1-\beta) T_{2})$. It is the closure found in \citeay{Lochon_PhdThesis_2016} among others.
\end{itemize}
We see that first these closures are a specific case where $F(\alpha_{2})$ is chosen to be zero in Equation~\eqref{eq:classic_cond_b}. Second, one could have chosen another interfacial velocity $\speed{sint}$ and it would have led to another interfacial pressure $\pI$ compatible with an entropy pair.

\begin{remark} If we had used the extended condition expressed in Equation~\eqref{eq:theo_b1b2_extended}, then the condition on $\partial_{\CV} \MES \left[ \CFnsym + \NCFnsym \right]$ would be
\begin{align}\label{eq:classic_cond_b_extended}
& \sum_{k = 1,2} \frac{1}{T_{k}} \left[p_{k}-\pI\left(\CV, \partial_{x} \CV \right) \right] \left[\speed{sint}\left(\CV, \partial_{x} \CV \right)-\speed{sk}\right] \partial_{x} \alpha_{k} \leq 0 \\
\Leftrightarrow \ & - \sum_{k = 1,2} \frac{1}{T_{k}} \frac{Z_{k}}{(Z_{1}+Z_{2})^{2}}\left[ p_{k^{\prime}}-p_{k} + sgn\left(\partial_{x} \alpha_{1} \right) (u_{k^{\prime}}-\speed{sk}) Z_{k^{\prime}} \right]^{2} \leq  0,
\end{align}
where $Z_{k}$ is defined by  $Z_{k} = \rho_{k} a_{k}$ with the phase sound speed $a^{2}_{k} = \left. \partial p_{k} / \partial \rho_{k}%
        \right|_{\scriptstyle s_{k}}$. From Equation~\eqref{eq:classic_cond_b_extended}, one sees that the dependency on $\partial_{x} \CV$ reduces to $\partial_{x} \alpha_{2}$ otherwise some terms would not be signable. Then closures such as the one found through Discrete Element Method (DEM) \citeay{Saurel_Gavrilyuk_2003} are obtained
\begin{align}
  \speed{sint} &= \frac{Z_{1} u_{1} + Z_{2} u_{2}}{Z_{1} +  Z_{2}} + sgn\left(\partial_{x} \alpha_{1} \right) \frac{p_{2}-p_{1}}{Z_{1}+Z_{2}}, \\
  \pI &= \frac{Z_{2} p_{1} + Z_{1} p_{2}}{Z_{1} +  Z_{2}} + sgn\left(\partial_{x} \alpha_{1} \right) \frac{Z_{1}Z_{2}}{Z_{1}+Z_{2}} \left(u_{2}-u_{1}\right).
\end{align}
\end{remark}
%
%
%To conclude, this first application of the method based on the general framework introduced in the previous section defines a systematic way of closing models along with a supplementary conservative equation and thus compatible thermodynamic. Moreover, it permits to recover all the closure of the interfacial terms found in the literature using a single method. This conforts also the choice of the decomposition made in Equations~\eqref{eq:decomp_a_b_applied}.
%
%
%
\subsection{Partially miscible phases entropy.}
Now, let us add a degree of freedom in the thermodynamics by introducing mixing terms in the definition of the entropy $\MES$ as in Equation~\eqref{eq:mixture_entropy_extended} to account for partial miscibility of the phases. The added terms, $\psi_{k}$, functions of the volume fraction $\alpha_{k}$ only, are to be determined.
%The methodology will raise conditions on $(\speed{sint}, \pI, \psi_{k})$ to close the model and the thermodynamic along with a supplementary conservative equation.

The entropic variables $\MEV$ are
\begin{align}\label{eq:v_r_ln_alpha}
  \MEV = \begin{pmatrix}
      \sum\limits_{\text{\tiny $k{=}1,2$}} (-1)^{k+1} \dfrac{p_{k}}{T_{k}} \left[ 1 - \dfrac{\alpha_{k}}{r_{k}} \psi_{k}^{\prime}(\alpha_{k}) \right]\\
      \MEV_{2} \\
      \MEV_{1}
    \end{pmatrix} %
  \text{ with } %
  \MEV_{k} = \frac{1}{T_{k}}\begin{pmatrix} 
      g_{k} - \speed{sk}^{2}/2\\
      \speed{sk} \\
      -1 
    \end{pmatrix}
\end{align}

%We now apply the conditions on $\partial_{\CV}(\ECF_{1}+\ECF_{2})$ and $\ENCF_{1} + \ENCF_{2}$ to determine $\transfer{salpha}(\CV)$ and derives thanks to computer algebra the equations that link the interfacial quantities $\speed{sint}$ and $\pI$ as well as the unknown mixing terms $\psi_{k}$.
%
%
\begin{theorem}\label{theo:BNZ_ab_mixing}
Consider System~\eqref{sys:BNZ_eq}. If the mixture entropy is defined as $\MES = - \sum_{k=1,2} \alpha_{k} \rho_{k} \left[ \PES_{k} - \psi_{k}(\alpha_{k}) \right]$ with $\psi_{k}$, $k=\left[ 1,2 \right]$, two strictly convex arbitrary functions depending on the volume fraction, then with the decomposition proposed in Equations~\eqref{eq:final_entropic_condition_decomp_BNZ}, we have
\begin{align}\label{eq:mixture_cond_a}
\begin{IEEEeqnarraybox}[\IEEEeqnarraystrutmode
\IEEEeqnarraystrutsizeadd{1pt}{1pt}][c]{rl}
  & \partial_{\CV \CV} \MES \times \partial_{\CV} \CF  +  \partial_{\CV} \transfer{v} \text{ symmetric } \\
  \Leftrightarrow \, &\transfer{salpha}(\CV) = \fontscal{F}(\alpha_{2}) + \frac{p_{1}}{T_{1}}u_{1} \left[ 1 - \frac{\alpha_{1}}{r_{1}} \psi_{1}^{\prime}(\alpha_{1}) \right] - \frac{p_{2}}{T_{2}}u_{2} \left[ 1 - \frac{\alpha_{2}}{r_{2}} \psi_{2}^{\prime}(\alpha_{2}) \right] 
\end{IEEEeqnarraybox}
\end{align}
with $\fontscal{F}$ a strictly convex arbitrary function depending on the volume fraction. As a consequence the condition on $\partial_{\CV} \MES \left[ \CFnsym + \NCFnsym \right]$ gives
\begin{align} \label{eq:mixture_cond_b}
\begin{IEEEeqnarraybox}[\IEEEeqnarraystrutmode
\IEEEeqnarraystrutsizeadd{1pt}{1pt}][c]{rcl}
  & \fontvec{0} &= \partial_{\CV} \MES(\CV) \left[ \CFnsym(\CV) + \NCFnsym(\CV) \right] \\
 \Leftrightarrow \ &   0  &= - \fontscal{F}(\alpha_{2}) + \sum_{k=1,2} (-1)^{k+1} \alpha_{k}\rho_{k} \psi_{k}^{\prime}(\alpha_{k}) (u_{k}-\speed{sint})\\%
   &  &+ \sum_{k=1,2} \frac{(-1)^{k}}{T_{k}} (\pI - p_{k} )( \speed{sk} - \speed{sint})
\end{IEEEeqnarraybox}
\end{align}
\end{theorem}
Again, Equation~\eqref{eq:mixture_cond_a} guarantees the existence of an entropy flux $\ETF(\CV)$ conditioning the function $\transfer{salpha}(\CV)$ $(Step\ 2.a)$. The interfacial quantities $(\speed{sint},\pI)$ and $\psi_{k}$ are linked by Equation~\eqref{eq:mixture_cond_b} $(Step \ 2.b)$.

The difference with the previous case for immiscible phases is that there are two supplementary unknowns $\psi_{k}$, $k=1,2$. We thus are free to either postulate first an interfacial velocity $\speed{sint}$ and then derive the corresponding $\pI$ and $\psi_{k}$ or postulate first the functions $\psi_{k}$ and see what choices we have for the interfacial terms. In the following we investigate the two approaches.

\subsubsection{Interfacial closures impacting thermodynamics.}

Let us postulate $\speed{sint}$ and limit ourselves to the case $\fontscal{F}(\alpha_{2}) =0$. We will again seek a linearly degenerate field for $\speed{sint}$. In such case, the results in Table~\ref{table:mixing_entropy_cond_case12} are obtained.
\begin{table}[ht]
\renewcommand{\arraystretch}{1,5}
\centering
    \caption{Admissible\small{ }thermodynamics and model closures obtained by postulating\small{ }$\speed{sint}$}
    \label{table:mixing_entropy_cond_case12}
    \begin{tabular}{c|c|c|c}
& $\speed{sint}$ & $\pI$ & $\left( \psi_{k} , \psi_{k^{\prime}}\right)$ \\ \hline \hline
Case 1 &$\speed{sk}$ & $p_{k^{\prime}}$ & $\left( \psi_{k}, 0 \right)$ \\ \hline
Case 2 & \specialcell{$\beta u_{1} + (1-\beta) u_{2}$\\ $\beta = \alpha_{1} \rho_{1} / \rho$} & \specialcell{$\mu p_{1} + (1- \mu ) p_{2}$ \\ $\mu\left(\beta\right) = \frac{(1-\beta) T_{2}}{\beta T_{1} + (1-\beta) T_{2}}$} & $\psi_{k}(\alpha_{k})=\psi_{k^{\prime}}(\alpha_{k^{\prime}})$
    \end{tabular}
\renewcommand{\arraystretch}{1}
\end{table}

In Case 1 of Table~\ref{table:mixing_entropy_cond_case12}, $\psi_{k}$ can be interpreted as a configuration energy of phase $k$ as in \citeay{Baer_Nunziato_1986}, \citeay{Kapila_1997} \citeay{Bdzil_1999}, in the context of deflagration-to-detonation. It is a term defining an interaction of one phase with itself only. More importantly, Equation~\eqref{eq:mixture_cond_b} shows that it is not possible to include a configuration energy for each phase when choosing the closure $(\speed{sint},\pI)=(\speed{sk},p_{k^{\prime}})$.

In Case 2 of Table~\ref{table:mixing_entropy_cond_case12}, the condition on the mixing term introduced in Equation~\eqref{eq:mixing_term_condition_BNZ} by \citeay{Gallouet_2004} is recovered and the closures are the one stated in \citeay{Coquel_2002}. However, the condition on the mixing terms imposes a constraint on the volume fraction and thus on the flow topology. Since mixing of the phases should be able to occur disregarding the flow topology, these terms fail to introduce free mixing among the phases.
\subsubsection{Thermodynamics impacting interfacial term closures.}
Since Case 1 and Case 2 of Table~\ref{table:mixing_entropy_cond_case12} do not allow the phases to mix, let us choose first the thermodynamics of the system and induce the admissible interfacial terms.

It has been shown that the mixing entropy of an ideal compressible binary mixture is of the form $\sum_{k=1,2} \alpha_{k} \text{ln}(\alpha_{k})$% \citeay{Gujrati_2003}
. Therefore, we choose to define the functions $\psi_{k}$ by $\psi_{k}(\alpha_{k}) =r_{k} \text{ln}(\alpha_{k})$. In this case, the entropy writes
\begin{align}
\MES = - \sum_{k=1,2} \alpha_{k} \rho_{k} \left[ \PES_{k} - r_{k}\text{ln}(\alpha_{k}) \right],
\end{align}
with $r_{k}$ the specific gas constant of phase $k$, we now account for quasi-miscibility between the phases.

The condition on $\transfer{salpha}$ degenerates, $\transfer{salpha} = \fontscal{F}(\alpha_{2})$ and the condition on $\partial_{\CV} \MES \left[ \CFnsym + \NCFnsym \right]$ is now
\begin{align}\label{eq:mixture_cond_b_rln}
\begin{IEEEeqnarraybox}[\IEEEeqnarraystrutmode
\IEEEeqnarraystrutsizeadd{1pt}{1pt}][c]{c}
  - \fontscal{F}(\alpha_{2}) + \pI \left( \frac{u_{1}-\speed{sint}}{T_{1}}-\frac{u_{2}-\speed{sint}}{T_{2}} \right) = 0.
\end{IEEEeqnarraybox}
\end{align}
It is no more possible to obtain the classic definition on $\speed{sint}$ and $\pI$. In the case $\fontscal{F}(\alpha_{2})=0$ two choices are possible to verify Equation~\eqref{eq:mixture_cond_b_rln} and summarized in Table~\ref{table:mixing_entropy_cond}.
\begin{table}[ht]
\renewcommand{\arraystretch}{1,5}
\centering
    \caption{Admissible\small{ }thermodynamics and model closures obtained by postulating\small{ }$\psi_{k}$}
    \label{table:mixing_entropy_cond}
    \begin{tabular}{c|c|c}
& $\speed{sint}$ & $\pI$ \\ \hline \hline
Case 3 & $\beta u_{1} + (1-\beta) u_{2}$ with $\beta = T_{2} / (T_{2}-T_{1})$ & no constraint \\
Case 4 & no constraint & 0
    \end{tabular}
\renewcommand{\arraystretch}{1}
\end{table}

Case 3 of Table~\ref{table:mixing_entropy_cond} proposes a temperature-based averaged velocity for $\speed{sint}$, which does not seem to be physically reasonable. In Case 4, the interfacial pressure must vanish for the system to admit an entropy supplementary conservation equation and the Baer-Nunziato model becomes a conservative system if one assumes the field associated to $\speed{sint}$ to be linearly degenerate. One knows how much it simplifies the problem in terms of numerical implementation. This result can be interpreted as an incompatibility between the existence of a mixing process in the thermodynamics of the mixture and an interfacial pressure, that stays meaningful as long as there is an interface between the two phases.
%
%
%\subsubsection{Towards a unified model for separated and dispersed flow}
\subsubsection{Link with dispersed phase flow.}
When the thermodynamics accounts for mixing (Case~4 Table~\ref{table:mixing_entropy_cond}), 
the existence of an entropy supplementary conservative equation is incompatible with the interfacial pressure, and thus 
%the interfacial pressure vanishes, cancelling 
the nozzling terms $\pI \partial_{x} \volfrac{k}$ vanish.

In separated two-phase flows, these terms are known to be necessary to preserve uniformity in velocity and pressure of the flow during its temporal evolution \citeay{Andrianov_2003} and are usually compared to the terms obtained in a single gas with a variable section \citeay{Saurel_2001}. Whereas these arguments seem valid for separated two-phase flows, one may question the role these terms play in a dispersed phase flows.

Taking the particular case $\pI = 0$ and $p_{2}=0$ in the Baer-Nunziato model seems to lead to a system of equations similar to one that would describe a flow of incompressible suspended particles, where 1 would denote the carrier phase and 2 the dispersed phase. Doing so, one recovers not only the Marble model \citeay{Marble_1963}, which proposes a pressureless gas dynamic equations for the
particle phase, valid in the limit where $\alpha_{2} < 10^{-3}$, but also the model obtained by Sainsaulieu \citeay{Sainsaulieu_1995} in the asymptotic limit where the volume fraction of the particles $\alpha_{2} \rightarrow 0$.

Nevertheless, even if the partial differential equations are alike, the thermodynamics associated to Marble and Sainsaulieu models differ from the one we propose for the Baer-Nunziato model. The latter accounts for compressibility of the two phases and partial miscibility whereas the thermodynamics of the Marble model assumes incompressibility of the particles and non-miscibility between the two phases.
% and the thermodynamics of the Sainsaulieu model assumes also incompressibility but accounts for stagnation pressure of the carrier phase on the particles.

To conclude, if one aims at unifying the description of both separated phases and dispersed flow through a unique model, the thermodynamics must be treated together with the system modelling.
%\newpage
\section{Application to the plasma model}\label{sec:plasma}
The multicomponent fluid modelling of plasma flows out of thermal equilibrium has been derived rigorously from kinetic theory using a multi-scale Chapman-Enskog expansion mixing a hyperbolic scaling for the heavy species with a parabolic scaling for the electrons \citeay{Graille_2007}. The system takes the form
\begin{align}\label{sys:plasma_eq_full}
  \partial_{t} \CV + \left[ \partial_{\CV} \CF(\CV) + \NCF(\CV) \right] \partial_{x} \CV = \partial_{x} \left( \DF(\CV) \partial_{x} \CV \right),
\end{align}
with
\renewcommand\arraystretch{1}
\setlength\arraycolsep{8pt}
\begin{align}\label{eq:plasma_eq}
\partial_{\CV} \CF(\CV) &= \begin{pmatrix}
        0 & 1 & 0 & 0 & 0 \\
        (\kappa/2-1)\speed{s}^2   & (2-\kappa)\speed{s} & \kappa & 0 & 0\\
        (\kappa/2 \speed{s}^2 - \frac{h^{tot}}{\rho_{h}})\speed{s}  & \frac{h^{tot}}{\rho_{h}} - \kappa \speed{s}^2 & (1+\kappa)\speed{s} & 0 & 0\\
        -\frac{\rho_{e}}{\rho_{h}} \speed{s} & \frac{\rho_{e}}{\rho_{h}} & 0 & \speed{s} & 0 \\
        - \frac{\rho_{e} \epsilon_{e}}{\rho_{h}} \speed{s} & \frac{\rho_{e} \epsilon_{e}}{\rho_{h}} & 0 & 0 & \speed{s}
    \end{pmatrix}, \\
\NCF(\CV) &= \begin{pmatrix}
        0 & 0 & 0 & 0 & 0\\
        0 & 0 & 0 & 0 & 0\\
        0 & 0 & 0 & 0 & 0\\
        0 & 0 & 0 & 0 & 0\\
        -\frac{\rho_{e} \epsilon_{e}}{\rho_{h}} \kappa \speed{s} & \frac{\rho_{e} \epsilon_{e}}{\rho_{h}} \kappa & 0 & 0 & 0
      \end{pmatrix}, \\
\DF(\CV) &= \begin{pmatrix}
        0 & 0 & 0 & 0 & 0\\
        0 & 0 & 0 & 0 & 0\\
        0 & 0 & 0 & -\frac{\lambda \kappa \epsilon_{e}}{\rho_{e}}  & \frac{\lambda \kappa \epsilon_{e}}{\rho_{e}} + \gamma D\\
        0 & 0 & 0 & 0 & \frac{D\kappa}{T_{e}}\\
        0 & 0 & 0 & -\frac{\lambda \kappa \epsilon_{e}}{\rho_{e}} & \frac{\lambda \kappa \epsilon_{e}}{\rho_{e}} + \gamma D
      \end{pmatrix},
\end{align}
where the column vector $\CV \in \spaces{R}^{5}$ is defined by $\CV^{\trans} = \left( \rho_{h}, \rho_{h} \speed{s}, E, \rho_{e}, \rho_{e} \epsilon_{e}\right)$ with $\rho_{h}$ is the density of the heavy particles, $\speed{s}$ the hydrodynamic velocity, $E$ the total energy defined by $E= \rho_{h} \speed{s}^2/2 + \rho_{h} \epsilon_{h} + \rho_{e} \epsilon_{e}$, $ \epsilon_{h}$ the internal energy of the heavy particles,  $\rho_{e}$ the density of the electrons, $\epsilon_{e}$ the internal energy of the electrons, $h^{tot}$ the total enthalpy defined by $h^{tot}= E + p$ with $p=p_{h}+p_{e}$, $T_{e}$ the temperature of the electrons, the constant $\kappa$ defined by $\kappa = \gamma-1$ with $\gamma$ the isentropic coefficient, $p_{h}$ is the pressure of the heavy particles and $p_{e}$ is the pressure of the electrons. In the diffusive terms, $\lambda$ is the electron thermal conductivity, D the electron diffusion coefficient.

Concerning the thermodynamics, it can be obtained from kinetic theory. The electrons and the heavy particles thermodynamics are defined by an ideal gas equation of state, and they share both the same isentropic coefficient: $p_{h} = \kappa \rho_{h} \epsilon_{h}$, $p_{e} = \kappa \rho_{e} \epsilon_{e}$ where $p_{h}$ is the pressure of the heavy particles and $p_{e}$ is the pressure of the electrons, $r$ is the constant of the gas $r=c_{v} \kappa$ with $c_{v}$ the calorific heat at constant volume, the model being adimensionalized $r=c_{v}(\gamma-1)=1$.

The model is naturally hyperbolic \citeay{Graille_2007} and also involves second-order terms and eventually source terms \citeay{Magin_2009}. Here we considered the homogeneous form.

In this section, we would like to derive the usual entropy supplementary conservative equation found by \citeay{Graille_2007} and show that it is unique, to attest the effectiveness of the theory.
\subsection{Decomposition.}
We need to proceed to the decomposition of the conservative and non conservative terms of System~\eqref{sys:plasma_eq_full}. We restrict ourselves again to the decomposition proposed in Definition~\eqref{def:decomposition_applied} and we add a degree of liberty to each non-null non-conservative components by defining $\transfer{v}: \CV \in \Omega \mapsto \spaces{R5}$ as $\transfer{v}(\CV)^{\trans} = ( \transfer{s1}(\CV), \transfer{s2}(\CV), 0, 0, 0)$ such that the following decompositions are obtained
\begin{align}\label{eq:final_entropic_condition_decomp_plasma}
\left( \partial_{\CV} \MES(\CV) \left[ \CFsym(\CV) + \NCFsym(\CV) \right] \right)^{\trans} = \MEV(\CV) \cdot \partial_{\CV} \CF(\CV) + \begin{pmatrix}
    \transfer{s1}(\CV)\\
    \transfer{s2}(\CV)\\
    0 \\
    0 \\
    0
\end{pmatrix}, \\
\left( \partial_{\CV} \MES(\CV) \left[ \CFnsym(\CV) + \NCFnsym(\CV) \right] \right)^{\trans} =\begin{pmatrix}
    - \transfer{s1}(\CV) - \frac{\rho_{e}}{\rho_{h}} \left(1-\frac{T_{e}}{T_{h}}\right) \speed{s} \\
    - \transfer{s2}(\CV) + \frac{\rho_{e}}{\rho_{h}} \left(1-\frac{T_{e}}{T_{h}}\right) \\
    0 \\
    0 \\
    0
\end{pmatrix}.
\end{align}
\normalsize
The unknown scalar functions $\transfer{v}_{k}(\CV)$ give the possibility to fractions of the non-conservative terms to be given to the matrix $\Fsymk$.
\subsection{Ideal Gas entropy.}
The entropy $\MES: \CV \in \Omega \mapsto \spaces{R}$ for two perfect gases is defined as
\begin{align}\label{eq:mixture_entropy_plasma}
  \MES = - \rho_{h} \PES_{h} - \rho_{e} \PES_{e},
\end{align}
with the partial entropies defined by
\begin{align}
  \PES_{h}= c_{v} \, \text{ln}\left(\frac{p_{h}}{\kappa \rho_{h}^{\transfer{v}}}\right), && %
  \PES_{e}= c_{v} \, \text{ln}\left(\frac{p_{e}}{\kappa \rho_{e}^{\transfer{v}}}\right).
\end{align}
This entropy includes mixing between the electrons and the heavy particles. Thus, we start applying our method $(Step \ 1)$ by postulating $\MES$ as in Equation~\eqref{eq:mixture_entropy_plasma}. The entropic variables $\MEV$ are then
\begin{align}
  \MEV = \begin{pmatrix}
      \dfrac{1}{T_{h}}\left(g_{h} - \speed{s}^{2}/2\right)\\
      \dfrac{1}{T_{h}}\speed{s}\\
      -\dfrac{1}{T_{h}}\\
      \dfrac{1}{T_{e}} g_{e}\\
      \dfrac{1}{T_{h}}-\dfrac{1}{T_{e}}
    \end{pmatrix},
\end{align}
with $g_{k}$ the Gibbs free energy, $g_{k} =\epsilon_{k} + p_{k}/\rho_{k} - T_{k}s_{k}$.
\begin{remark} In the fourth component of the entropic variable, the kinetic energy of the electrons has vanished. This is due to the low-Mach assumption made for the electrons.
\end{remark}
 We now apply the conditions to determine $\transfer{sk}(\CV)$.
\begin{theorem}\label{theo:plasma_ab_classic}
Consider System~\eqref{sys:plasma_eq_full}. If the mixture entropy is defined as $\MES = - \rho_{h} \PES_{h} - \rho_{e} \PES_{e}$, then with the decomposition proposed in Equations~\eqref{eq:final_entropic_condition_decomp_plasma}, we have
\begin{align}\label{eq:plasma_classic_cond_a}
\begin{IEEEeqnarraybox}[\IEEEeqnarraystrutmode
\IEEEeqnarraystrutsizeadd{1pt}{1pt}][c]{rl}
  & \partial_{\CV \CV} \MES(\CV) \times \partial_{\CV} \CF(\CV) + \partial_{\CV} \transfer{v}(\CV) \text{ symmetric } \\
  \Leftrightarrow \  & \transfer{s1} (\CV) = \frac{\rho_{e}}{\rho_{h}}\left(1-\frac{T_{e}}{T_{h}}\right)\speed{s} \text{ and } \transfer{s2} (\CV) = -\frac{\rho_{e}}{\rho_{h}}\left(1-\frac{T_{e}}{T_{h}}\right)
\end{IEEEeqnarraybox},
\end{align}
and the condition on $\partial_{\CV} \MES(\CV) \left[ \CFnsym(\CV) + \NCFnsym(\CV) \right]$ is
\begin{align}\label{eq:plasma_classic_cond_b}
  \partial_{\CV} \MES(\CV) \left[ \CFnsym(\CV) + \NCFnsym(\CV) \right] & = (0,\, 0,\, 0,\, 0,\, 0).
\end{align}
\end{theorem}
\begin{proof} Using Maple{\texttrademark}, we find
\begin{align*}
  &\partial_{\CV \CV} \MES(\CV) \times \partial_{\CV} \CF(\CV) + \partial_{\CV} \transfer{v}(\CV) \text{ symmetric } \nonumber \\
  \Leftrightarrow & \transfer{s1} (\CV) = \frac{\rho_{e}}{\rho_{h}}\left(1-\frac{T_{e}}{T_{h}}\right)\speed{s} + \int \left[ -\speed{s} \partial_{\speed{s}}\fontscal{F}_{1}(\rho_{h},\speed{s}) +\rho_{h} \partial_{\rho_{h}}\fontscal{F}_{1}(\rho_{h},\speed{s}) \right] d\speed{s} + \fontscal{F}_{2}(\rho_{h})  \nonumber \\
         & \text{and } \transfer{s2} (\CV) = -\frac{\rho_{e}}{\rho_{h}}\left(1-\frac{T_{e}}{T_{h}}\right) + \fontscal{F}_{1}(\rho_{h},\speed{s}),
\end{align*}
with $\mathsf{F_{1}}$, $\mathsf{F_{2}}$ two arbitrary functions and the condition on $\partial_{\CV} \MES(\CV) \left[ \CFnsym(\CV) + \NCFnsym(\CV) \right]$ is
\begin{align*}
  \left( \partial_{\CV} \MES(\CV) \left[ \CFnsym(\CV) + \NCFnsym(\CV) \right] \right)^{\trans} & = \begin{pmatrix}
    - \int \left[ -\speed{s} \partial_{\speed{s}}\fontscal{F}_{1}(\rho_{h},\speed{s}) +\rho_{h} \partial_{\rho_{h}}\fontscal{F}_{1}(\rho_{h},\speed{s}) \right] d\speed{s} - \fontscal{F}_{2}(\rho_{h}) \\
    -\fontscal{F}_{1}(\rho_{h},\speed{s}) \\
    0\\
    0\\
    0
  \end{pmatrix}\\
   &= \bs{0}.
\end{align*}
One sees that the last equation imposes first $\fontscal{F}_{1}=0$ and thus $\fontscal{F}_{2}=0$. Reinjecting these terms into the first equation gives the result.
\end{proof}

As explained in $(Step\ 2.a)$, the Equation~\eqref{eq:plasma_classic_cond_a} guarantees the existence of an entropy flux $\ETF:\CV \in \Omega \mapsto \spaces{R}$ associated with the entropy $\MES$ defined in Equation~\eqref{eq:mixture_entropy_plasma} by solving the unknown functions $\transfer{s1}(\CV)$ and $\transfer{s2}(\CV)$.

Therefore, for the entropy $\MES$ defined in Equation~\eqref{eq:mixture_entropy_plasma}, there is a unique decomposition which ensures the existence of a supplementary conservative equation which is given by
\begin{align}\label{eq:final_entropic_condition_decomp_plasma_determined}
\left( \partial_{\CV} \MES \left[ \CFsym + \NCFsym \right] \right)^{\trans}= \MEV^{\trans} \cdot \partial_{\CV} \CF(\CV) + \begin{pmatrix}
    \frac{\rho_{e}}{\rho_{h}}\left(1-\frac{T_{e}}{T_{h}}\right)v\\
    \frac{\rho_{e}}{\rho_{h}}\left(1-\frac{T_{e}}{T_{h}}\right)\\
    0 \\
    0 \\
    0
\end{pmatrix},
\end{align}
\begin{align}
\partial_{\CV} \MES \left[ \CFnsym + \NCFnsym \right] = \bs{0}.
\end{align}
It leads to the following entropy flux couple
\begin{align}
  \MES &= - \rho_{h} \PES_{h} - \rho_{e} \PES_{e}, \\
  \ETF &= - \left( \rho_{h} \PES_{h} + \rho_{e} \PES_{e} \right) \speed{s}.
\end{align}
The theory recovers the supplementary conservative equation already found in the literature from the kinetic theory \citeay{Graille_2007}.
%%\newpage
%\input{mixture_entropy}
%\newpage
%\input{unified_two_phase_flow_model}
%\newpage
\section{Conclusion}
In the present contribution, we have proposed a theoretical framework for the derivation of supplementary conservation laws for systems of partial differential equation including  first-order  non-conservative terms - commonly encountered in modeling of complex flows - thus extending the standard approach for systems of conservation laws. Since our main objective is deriving 
an  entropy supplementary conservation law, we have used this framework to make a first step to 
extend the theory of Godunov-Mock to such non-conservative systems.

Given a reasonable choice in the combination of the conservative and non-conservative terms, we have been able to show how to use the theory to design or analyze systems by means of computer algebra on two applications chosen for their numerous differences in terms of model and thermodynamics closure as well as the nature of the waves impacted by the non-conservative terms.

Firstly, applied to the Baer-Nunziato two-phase flow model derived from rational thermodynamics, the theory has brought about entropy supplementary conservative equations together with constraints on the interfacial quantities and the definition of the thermodynamics for non-miscible fluids and also when accounting for some level of mixing of the two phases. A new closure for the interfacial quantities has been proposed and leads to a conservative system.
Secondly, for a plasma model obtained rigorously from the kinetic theory of gases, where the thermodynamics is also provided, the approach allows to recover as unique the supplementary conservation equation related to the kinetic entropy and is thus assessed.
%
%A strongly connected question for such systems is the ability to derive an entropic symmetrization in the sense of Godunov-Mock and the related constraints on the decomposition, as well as the study of the spectrum and hyperbolicity. The proposed framework introduced in this paper allows to shed some light on these questions and the two systems, we have applied the theory to, are especially suited for such a purpose. Nevertheless, in the case of Baer-Nunziato model, the lack of strict convexity of the proposed entropies prevents its Godunov-Mock symmetrization. This loss of strict convexity in the framework of non-interacting thermodynamics has been investigated in \citeay{Cordesse_CMT_2019} where a mixing thermodynamics for multi-fluid has been developed. Based on this new developments, we hope that equipping the Baer-Nunziato system with an extended thermodynamics closure will lead to a strictly convex entropy and thus to its symmetrization in the sense of Godunov-Mock. This is the subject of our current research.

The content of the paper is a first step into studying the entropic symmetrization in the sense of Godunov-Mock and relation to source terms for two-phase flow modeling. Some partial symmetrization of the Baer-Nunziato model has been obtained in the classical framework by \cite{Forestier_2011}. Combining such symmetrization theory with source terms can then be envisioned such as in the case of plasma flows \cite{Magin_2009}, even if the symmetrization is only partial in the framework of \cite{Graille_2007} where the electron are considered in a low-Mach limit. Nevertheless, for such a study to be complete, several other steps have to be handled first: the question of the strict convexity of the entropy for the change of variable to be admissible and its relation to thermodynamics (a difficult question \cite{Coquel_2002,Gallouet_2004}); it is a part of Pierre Cordesse's PhD Thesis \cite{Cordesse_PhD}. This loss of strict convexity in the framework of non-interacting thermodynamics has been investigated in \citeay{Cordesse_CMT_2019} where a mixing thermodynamics for multi-fluids has been developed. Based on this new developments, we hope that equipping the Baer-Nunziato system with an extended thermodynamics closure will lead to a strictly convex entropy and thus allow the study of entropic full symmetrization and source terms, in the spirit of \cite{Giovangigli_1998,Massot_2002,Giovangigli_2004,Magin_2009}. This is the subject of our current research.

%\newpage

%\appendix

%\input{appendices}

\section*{Acknowledgment}
The authors would like to acknowledge the support of a CNES/ONERA PhD Grant for P. Cordesse and the help of M. Th\'eron (CNES). They would like to express their special thanks to F. Coquel, S.~Kokh, V.~Giovangigli and A.~Murrone for their invaluable help and numerous pieces of advice during the writing of the paper. We also would like to thank discussions with J.M. H\'erard, which prompted this research path. Part of this work was conducted during the Summer Program 2018 at NASA Ames Research Center and the support and help of Nagi N. Mansour is also gratefully acknowledged.
%
%\nocite{*}
\printbibliography
%\vfill % Fill the rest of the page with whitespace
%
%
\end{document}